\documentclass[10pt, doublecolumn]{IEEEtran}
\usepackage{epsfig,latexsym}
\usepackage{float}
\usepackage{indentfirst}
\usepackage{amsmath}
\usepackage{bm}
\usepackage{amssymb}
\usepackage{times}

\usepackage{algorithm}
\usepackage[noend]{algpseudocode} 
\usepackage{subfigure}
\usepackage{psfrag}
\usepackage{hyperref}
\usepackage{cite}
\usepackage{lastpage}
\usepackage{fancyhdr}
\usepackage{color}
 \usepackage{amsthm}
\usepackage{bigints}
\usepackage{booktabs}
\usepackage{multirow}
\usepackage{siunitx}

\newtheorem{Lemma}{Lemma}
\newtheorem{lemma}[Lemma]{$\mathbf{Lemma}$}

\newcounter{problem}
\newcounter{save@equation}
\newcounter{save@problem}
\makeatletter
\newenvironment{problem}
{\setcounter{problem}{\value{save@problem}}%
  \setcounter{save@equation}{\value{equation}}%
  \let\c@equation\c@problem
  \subequations
}
{\endsubequations
  \setcounter{save@problem}{\value{equation}}%
  \setcounter{equation}{\value{save@equation}}%
}

\begin{document}
\title{  \huge{BackCom Assisted Hybrid NOMA Uplink Transmission for Ambient IoT  }}

\author{\author{ Zhiguo Ding, \IEEEmembership{Fellow, IEEE}  \thanks{ 
  
\vspace{-2em}

  Z. Ding is with  Khalifa University, Abu Dhabi, UAE, and  University of Manchester, Manchester, M1 9BB, UK.      }
 
 \vspace{-2em}

  }\vspace{-4em}}
 \maketitle

\vspace{-1em}
\begin{abstract}
Hybrid non-orthogonal multiple access (H-NOMA)  has recently received significant attention as a general framework of multiple access, where both conventional orthogonal multiple access (OMA) and pure NOMA are its special cases. This paper focuses on the application of H-NOMA to ambient Internet of Things (IoT) with energy-constrained devices, where a new backscatter communication (BackCom) assisted H-NOMA uplink scheme is developed.  Resource allocation for H-NOMA uplink transmission is also considered, where an overall power minimization problem is formulated. Insightful understandings for the key features of BackCom assisted H-NOMA and its difference from conventional H-NOMA are illustrated by developing analytical results for the two-user special case. For the general multi-user scenario, two algorithms, one based on the branch-bound (BB) principle and the other based on successive convex approximation (SCA), are developed to realize different tradeoffs between the system performance and complexity. The numerical results are also provided to verify the accuracy of the developed analytical results and demonstrate the performance gain of H-NOMA over OMA.  
\end{abstract}\vspace{-0.5em}

\vspace{-1em} 

\section{Introduction}
Non-orthogonal multiple access (NOMA) has been recently identified as a crucial  component of  the International Mobile Telecommunications (IMT)-2030 Framework issued by the International Telecommunication Union (ITU) \cite{imt2030}. The key feature of NOMA is its superior spectral efficiency, where the use of NOMA can ensure  that the scarce bandwidth resources are efficiently shared among multiple users, instead of being occupied by a single user \cite{you6g,mojobabook}. Another key feature of NOMA is its excellent compatibility  with other communication techniques. For example, the use of NOMA has been shown to be crucial for supporting  secure offloading in multi-access edge computing systems \cite{10363117}. In addition, the principle of NOMA is naturally applicable to integrated sensing and communications (ISAC), as shown in \cite{10024901}. Furthermore,  the application of NOMA has been shown to be important to   support   short-packet transmission in fluid antenna systems     \cite{10423153},  enhance the user fairness in reconfigurable intelligent surface (RIS) networks \cite{10311519}, support network slicing in cell-free systems \cite{10278091}, and realize integrated terrestrial-satellite communications \cite{7972935}. 

This paper is to focus on the compatibility between NOMA and conventional orthogonal multiple access (OMA) \cite{Zhiguo_CRconoma, 9693417,8823023}. This compatibility is important, not only because the next-generation multiple access  is envisioned to be a general framework of NOMA and OMA, but also because a NOMA form      compatible to OMA can be immediately deployed in the existing OMA networks without the need to modify the current wireless standards  \cite{9352956,9693536,8641304}. In particular, this paper focuses on hybrid NOMA (H-NOMA) which offers two benefits  \cite{hnomadown,9340353,9679390}. One benefit is that the implementation of H-NOMA does not cause any disruption to the existing OMA network. Take time division multiple access (TDMA) as an example of OMA. By using H-NOMA,   a user   still finishes  its transmission by the end of its own TDMA   time slot, and   also has  access to multiple time slots, instead of one as in TDMA. The other benefit is that   H-NOMA enables        multiple-time-slot resource allocation which yields better performance than conventional  single-time-slot resource allocation. 

Recall that one of the key applications of NOMA is massive machine type communication (mMTC), and mMTC  is uniquely featured by energy constrained devices, such as Internet of Things (IoT) sensors and smart-home   devices.  This energy constraint feature has motivated the combination of NOMA with backscatter communications (BackCom) which realizes the following two types of cooperation among multiple users \cite{9261963, backnoma}.  The first type is the spectrum cooperation, i.e., the spectrum is more efficiently  shared among  the users with diverse quality of service requirements, such as virtual reality (VR) users and IoT sensors. The second type is the energy cooperation among the users with different energy profiles, e.g.,  users with energy supplies can help those energy constrained sensors via BackCom \cite{9706269, 10244014,10412664}.  We note that the idea of BackCom assisted NOMA is closely related to symbiotic radio \cite{8907447,8636518}, since the two concepts are based on the same idea that one user employs  another user's signal as carrier signals via BackCom. However, to the author's best knowledge, the concept of H-NOMA has not been applied to BackCom assisted NOMA or symbiotic radio in the literature, which motivates this paper. 

The aim of this paper is to develop a new BackCom assisted H-NOMA uplink transmission scheme and also provide the performance analysis for revealing the key features of this new transmission scheme.  The contributions of the paper are listed as follows:
\begin{itemize}
\item By assuming the existence of a multi-user  TDMA legacy network, a new BackCom assisted H-NOMA uplink transmission scheme is developed. This H-NOMA scheme is a general transmission framework, where a few well-known transmission strategies, such as pure OMA and pure NOMA, are its special cases. Recall that the key feature of BackCom assisted NOMA is that one user's information bearing signal can be employed  as the carrier signal by other users. By using this feature, a successive interference cancellation (SIC)   strategy  tailored to BackCom assisted H-NOMA is developed, which makes the expressions of the  data rates achievable to BackCom assisted H-NOMA   different from those of     conventional H-NOMA   shown in \cite{hnomadown,9679390,9340353}. A challenging non-convex power minimization problem is also formulated for BackCom assisted H-NOMA, where how different choices of power allocation lead to pure OMA, pure NOMA and H-NOMA are illustrated. 

\item The unique features of BackCom assisted H-NOMA and its difference from   conventional H-NOMA are revealed by developing the analytical results for   the two-user special case. Recall for conventional H-NOMA uplink \cite{9679390}, if the durations of time slots are  same, the pure OMA mode is always optimal, i.e., pure OMA outperforms both pure NOMA and H-NOMA. Our developed analytical results show that for BackCom assisted H-NOMA, pure OMA can never be optimal, i.e., either pure NOMA or H-NOMA can outperform OMA. Another conclusion  made to conventional H-NOMA uplink is that the pure NOMA mode is never adopted. For the considered BackCom assisted H-NOMA scheme, the analytical results are developed to show that pure NOMA could be optimal, i.e., it is possible for pure NOMA to outperform both pure OMA and H-NOMA. These conclusions  reveal the key difference between BackCom assisted H-NOMA and   conventional H-NOMA, and they are  due to the fact that with  BackCom,  NOMA transmission does not cause  extra energy consumption to the users.    

\item For the general multi-user scenario,  two algorithms are developed to realize different tradeoffs between the system performance and complexity. First, to tackle the non-convexity of the formulated power minimization problem, a  successive convex approximation (SCA) algorithm is developed to yield a low-complexity sup-optimal solution \cite{6365845,7946258}.  Second, to identify the lower bound on the overall power consumption required by BackCom assisted H-NOMA, a branch and bound (BB) algorithm is also developed to obtain the optimal solution of the formulated power minimization problem \cite{5765556,8170332}. Recall that  the key step of a BB algorithm is the feasibility check. An intuitive approach is to again apply SCA for the feasibility check; however, our simulation results show that this SCA based BB algorithm achieves  poor performance and results in high computational complexity. A successive resource allocation based feasibility approach is developed by using the successive nature of SIC.  Simulation results are provided to demonstrate the superior performance gain of H-NOMA over conventional OMA, and verify the accuracy of the developed analytical results. 
\end{itemize}
   
   The remainder  of this paper is organized as follows. In Section \ref{section 2}, the BackCom assisted H-NOMA uplink scheme is developed, and an overall power minimization problem is formulated.   In Sections \ref{section 3}, insightful understandings about BackCom assisted H-NOMA are obtained by focusing on the two-user special case.   In Section \ref{section 4}, two algorithms based BB and SCA are developed for carrying out resource allocation in the general multi-user case.  In Section \ref{section 5}, the performance of BackCom assisted H-NOMA is evaluated by using computer simulation results,    and  Section \ref{section 6} concludes the paper.  
   
   \section{System Model}\label{section 2}
Consider an uplink communication scenario with $M$ users, denoted by ${\rm U}_m$, communicating with the same base station, where each node is     equipped with a single antenna. In addition, it is assumed that each user is equipped with   a BackCom circuit. Furthermore, assume that the users are ordered according to their channel gains, i.e., $|h_1|^2< \cdots < |h_M|^2$, where $h_m$ denotes ${\rm U}_m$'s channel gain. The rationale behind this channel ordering will be discussed in Section \ref{ssubse}.

\subsection{A TDMA Based Legacy Network}
With TDMA, the $M$ users take turns to send their information to the base station. Without loss of generality, assume that    ${\rm U}_m$ is served in the $m$-th time slot.  With OMA, in the $m$-th time slot, ${\rm U}_m$ receives the following signal: $y^{\rm OMA}_m = h_m\sqrt{P_m^{\rm OMA}}s_m+n_m$, where   $P_m^{\rm OMA}$ denotes ${\rm U}_m$'s transmit power in OMA,  $s_m$ denotes ${\rm U}_m$'s unit-power signal, and $n_m$ denotes additive white Gaussian noise. Therefore, ${\rm U}_m$'s data rate in OMA is simply given by $R^{\rm OMA}_m=\log\left(1+ |h_m|^2P_m^{\rm OMA}\right)$, where   the noise power is assumed to be normalized. It is straightforward to show that in order to achieve   the target data rate, denoted by $R$, ${\rm U}_m$'s transmit power, $P_m^{\rm OMA}$, needs to be $P_m^{\rm OMA}=\frac{e^R-1}{|h_m|^2}$, and hence the minimal total power consumption in OMA is given by
\begin{align}\label{oma solution}
P^{\rm OMA} = \sum^{M}_{m=1}\frac{e^R-1}{|h_m|^2}.
\end{align}

\subsection{BackCom Assisted H-NOMA}\label{ssubse}
 The aim of this paper is to implement BAC-NOMA as an add-on in the aforementioned  TDMA legacy network. 
In particular, during the first time slot which   belongs to ${\rm U}_1$ in TDMA, ${\rm U}_m$, $2\leq m \leq M$, also transmits its signal to the base station by modulating and reflecting the signal from ${\rm U}_1$. As a result, during the first time slot, the base station receives the following signal:
\begin{align}
y_1 = \sqrt{P_1}h_1s_{11}+\sum^{M}_{m=2}h_mg_{m1}s_{11}\sqrt{\eta_{m1}P_1}s_{m1}+n_1
\end{align}
where   $g_{mi}$ denotes the channel gain between ${\rm U}_i$ and ${\rm U}_m$, $s_{mi}$ denotes the unit-power signal sent by ${\rm U}_m$ in the $i$-th time slot, $\eta_{mi}$ denotes ${\rm U}_m$'s reflecting coefficient in the $i$-th time slot, $\eta_{mi}\leq 1$, $P_m$ denotes ${\rm U}_m$'s transmit power in the $m$-th time slot,  and $n_m$ denotes the corresponding additive white Gaussian noise. 

The base station needs to carry out SIC to decode the users' signals. The SIC strategy adopted in this paper is to decode  ${\rm U}_1$'s signal   first, since its signal can be viewed as a carrier signal and required to be known when the other users' signals are decoded. 
Therefore, during the first time slot, ${\rm U}_1$'s signal can be decoded with the following data rate:
\begin{align}\label{eqx1}
R_{11} = \log\left(
1+\frac{P_1|h_1|^2}{\sum^{M}_{m=2}|h_m|^2|g_{m1}|^2  {\eta_{m1}P_1} +1}
\right),
\end{align}
where the noise power is assumed to be normalized. Similarly, after ${\rm U}_{m-1}$'s signal is decoded  successfully,  ${\rm U}_m$'s   signal can be decoded in the first time slot with the following data rate:
\begin{align}\label{eqx2}
R_{m1} = \log\left(
1+\frac{|h_m|^2|g_{m1}|^2  {\eta_{m1}P_1} }{\sum^{M}_{j=m+1}|h_j|^2|g_{j1}|^2  {\eta_{j1}P_1} +1}
\right).
\end{align}
The data rate expressions in \eqref{eqx1} and \eqref{eqx2} show that ${\rm U}_{m}$ will suffer more interference  than ${\rm U}_{i}$, $m<i$, which motivates the assumption that the users are ordered according to their channel conditions in a descending order. 
 
In general, the use of BackCom assisted H-NOMA can ensure that during the $i$-th time slot, ${\rm U}_m$ can achieve the following data rate:
\begin{align}
R_{mi} = \log\left(
1+\frac{|h_m|^2|g_{mi}|^2  {\eta_{mi}P_i} }{\sum^{M}_{j=m+1}|h_j|^2|g_{ji}|^2  {\eta_{ji}P_i} +1}
\right).
\end{align}

{\it Remark 1:} The main advantage of the proposed BackCom assisted H-NOMA scheme is its ambient feature, i.e.,   the additional NOMA transmissions are in a `battery-less' manner. For example, during the first time slot, ${\rm U}_2$'s transmission does not consume extra energy; instead, it relies on the existing OMA transmission carried out by ${\rm U}_1$, which is different from    those conventional  H-NOMA schemes   \cite{hnomadown,9679390}. 

However, the use of H-NOMA transmission may potentially increase the overall transmit power. For example,  in OMA,  ${\rm U}_1$ can solely occupy the first time slot. However, with BAC-NOMA, the signals reflected by ${\rm U}_m$, $2\leq m \leq M$, cause   interference to ${\rm U}_1$, which means that ${\rm U}_1$ might need to use more transmit power than that in OMA.  Motivated by this observation, the aim of this paper is to consider the following power minimization problem: 
  \begin{problem}\label{pb:1} 
  \begin{alignat}{2}
\underset{P_{m}, \eta_{mi}\geq0  }{\rm{min}} &\quad   \sum^{M}_{m=1} P_m \label{1tst:1}
\\ s.t. &\quad   \sum^{m}_{i=1}R_{mi}T\geq N,    1\leq m \leq M \label{1tst:2} 
\\ &\quad  \eta_{mi}\leq 1,   1\leq i \leq m-1,  1\leq m \leq M,\label{1tst:2} 
  \end{alignat}
\end{problem} 
where $T$ denotes the duration of each time slot, and $N$ denotes the target number of nats to be sent within $M$ time slots. 

{\it Remark 2:} Compared to the power minimization problem considered for conventional H-NOMA in \cite{9679390}, problem \eqref{pb:1} is more challenging. In particular, successive resource allocation which is shown to be Pareto optimal to conventional H-NOMA in \cite{9679390} cannot be applied. The reason for this challenge is that the users' individual decisions for their transmit powers and reflecting coefficients are strongly  coupled. For example, for conventional H-NOMA, one user's signal  is always treated as the interference to the other users, i.e., one user's transmit power   always appears in the denominator of the other users' signal-to-interference-plus-noise ratios (SINRs). However, due to the use of BackCom, one user's signal becomes useful to the other users, i.e., one user's signal is used as the carrier signal by the other users, and hence    one user's transmit power can also appear in the numerator of the other users' SINRs. Such coupling makes the formulated optimization problem challenging to solve.  

{\it Remark 3:} The proposed   BackCom assisted H-NOMA scheme is a general framework of OMA and NOMA. For example, by setting $\eta_{mi}=0$, for $1\leq i \leq m-1$ and $1\leq m \leq M$, the H-NOMA scheme is degraded to the pure OMA mode. More details about how different choices of $P_m$ and $\eta_{mi}$ lead to different modes will be provided  in the following section.   An intuition  for the choices of the reflecting coefficients is that each user should set $\eta_{mi}$ as $1$, i.e., its reflection should be  in a `full blasting' manner,  since a user's NOMA transmission does not consume its own battery. However, the insight obtained from the following two-user special case      shows that this intuition  is not true. 
\section{Performance Analysis for The Two-User Special Case}  \label{section 3}
In this section, the two-user special case of problem \eqref{pb:1} is focused on, where the obtained analytical results offer insight not only to the difference between   BackCom assisted H-NOMA and the conventional ones, but also to the key features of the considered optimization problem. 

With $M=2$, problem \eqref{pb:1} can be simplified as follows: 
  \begin{problem}\label{pb:2} 
  \begin{alignat}{2}
\underset{   }{\rm{min}} &\quad   P_1+P_2 \label{2tst:1}
\\ s.t. &\quad      \log\left(
1+|h_2|^2|g_{21}|^2  \eta_{21}P_1 
\right)+ \log\left(
1+ |h_2|^2P_2
\right)\geq R\label{2tst:2}
\\&\quad 
   \log\left(
1+\frac{|h_1|^2 P_1 }{ |h_2|^2|g_{21}|^2  {\eta_{21}P_1} +1}
\right)\geq R\label{2tst:3}
\\&\quad  \eta_{21}\geq 0, P_1\geq 0, P_2\geq 0,
  \end{alignat}
\end{problem} 
wher  $R=\frac{N}{T}$. Problem \eqref{pb:2} is not a convex optimization problem due to the coupling effects between the transmit powers and the reflecting coefficient. 

By defining  $\gamma_1=|h_1|^2$, $\gamma_2=|h_2|^2$,  $\gamma_0=|h_2|^2|g_{21}|^2$, and $P_0=\eta_{21}P_1$,  problem \eqref{pb:2} can be   simplified as follows:
 \begin{problem}\label{pb:3} 
  \begin{alignat}{2}
\underset{P_{i}\geq0  }{\rm{min}} &\quad   P_1+P_2 \label{3tst:0}
\\  s.t. &\quad       R- \log\left(
1+\gamma_0  P_0 
\right)-\log\left(
1+ \gamma_2P_2
\right)\leq 0\label{3tst:1}
\\&\quad  \label{3tst:2}
  \epsilon \gamma_0 P_0+\epsilon-\gamma_1 P_1\leq 0
 \\&\quad P_0\leq P_1, \label{3tst:3}
  \end{alignat}
\end{problem} 
where $\epsilon=e^R-1$. It is straightforward to show that problem \eqref{pb:3} is a convex optimization problem. However, compared to conventional  H-NOMA, the use of BackCom makes the analysis   of the optimal solution   more complicated. In particular,        the optimal solution of   problem \eqref{pb:3} could be in multiple forms,   as listed in Table \ref{table1}, where $f_{\rm (P3c)}$ is the function at the left hand side  of \eqref{3tst:2}. 

\begin{table}[t]\vspace{-1em}
  \centering
  \caption{ Potential Solutions of Problem \eqref{pb:3} }
  \begin{tabular}{|c|c|c|c|c|c|}
\hline
&    Pure&Pure  & H-NOMA &H-NOMA &H-NOMA \\&  OMA   &  NOMA& Type I&Type II&Type III \\
    \hline
   $P_1$      &$\neq0$	&$\neq0$ &$\neq0$ &$\neq 0$  &$\neq 0$  
    \\
    \hline
    $P_2$     &$\neq0$ &
    			  $=0$ &	$\neq0$  	&$\neq0$ &$\neq 0$     \\
    \hline
    $P_0$     &$=0$ &
    			  $\neq0$ &	$\neq P_1$  	&$= P_1$   &$= P_1$ \\
    \hline
    $f_{\rm (P3c)}$     &  &
    			    &	   	&$\neq 0$ &$=0$ \\\hline
  \end{tabular}\vspace{1em}\label{table1}\vspace{-1em}
\end{table}  

To facilitate the performance analysis of BackCom assisted H-NOMA, the following lemma is provided to show the optimality of the pure OMA solution. 
\begin{lemma}\label{lemma1}
 The pure OMA solution shown in \eqref{oma solution} cannot be the optimal solution of problem \eqref{pb:3}.  
\end{lemma}
\begin{proof}
See Appendix \ref{proof1}. 
\end{proof}
{\it Remark 3:} Lemma \ref{lemma1} reveals the first interesting feature of BackCom assisted H-NOMA. Recall that for conventional H-NOMA, if the durations of time slots are equal, the OMA solution is   optimal and outperforms the NOMA solutions \cite{9679390}. Lemma \ref{lemma1} illustrates that OMA can never be optimal for   problem \eqref{pb:3},  a key difference between BackCom assisted H-NOMA and conventional H-NOMA.

Table \ref{table1} shows that pure NOMA is also a potential transmission strategy. It is important to point out that pure NOMA cannot be optimal for conventional H-NOMA, as illustrated in the following lemma.

\begin{lemma}\label{lemma2}
For the power minimization problem of conventional H-NOMA, i.e., problem \eqref{pb:3} with $P_0+P_1+P_2$ as its objective function and without the constraint $P_0\leq P_1$, the pure NOMA solution cannot be optimal. 
\end{lemma}
\begin{proof}
See Appendix \ref{proof2}.
\end{proof}

For BackCom assisted H-NOMA, the following conclusion can be made to pure NOMA.

\begin{lemma}\label{lemma3}
Pure NOMA is a potentially optimal transmission strategy for BackCom assisted H-NOMA.   In particular, if $ \frac{\gamma_2} { \gamma_0}\leq \frac{1}{1+\epsilon}$ and $  \quad \frac{ \gamma_1}{\gamma_0}\geq  (1+\epsilon)$,   $P_0^*=P_1^*=\frac{\epsilon}{\gamma_0}$ and $P_2^*=0$. If  $ \frac{\gamma_2}{\gamma_1} \leq \frac{1}{\epsilon (1+\epsilon)}$ and $ \frac{ \gamma_1}{\gamma_0}\leq  (1+\epsilon)$, $P_0^*=\frac{\epsilon}{\gamma_0}$,  $P_1^*=\frac{\epsilon(1+\epsilon)}{\gamma_1}$ and $P_2^*=0$. 
\end{lemma}
\begin{proof}
See Appendix \ref{proof3}.
\end{proof}
{\it Remark 4:}   Lemma \ref{lemma3}  illustrates another key difference between BackCom assisted H-NOMA and conventional H-NOMA. In particular,   pure NOMA cannot be optimal to conventional H-NOMA, as shown in Lemma \ref{lemma2}, but could be   an optimal strategy in BackCom assisted H-NOMA networks, as shown in Lemma \ref{lemma3}.  

{\it Remark 5:} Lemma \ref{lemma3} also reveals how complicated it is to obtain a closed-form expression for the optimal solution of problem \eqref{pb:3}. Even for a simple transmission strategy like pure NOMA, there are two possible expressions for the optimal power allocation. For a more complicated transmission strategy, such as H-NOMA, the expression of the corresponding optimal power allocation is more complicated, as shown in the following lemma.    

\begin{lemma}\label{lemma4}
For the considered power minimization problem, there are three possible forms of the H-NOMA solution, as listed in the following:
\begin{align}
 P_0^{\rm I}&= \frac{\gamma_1 \epsilon_1 }{\epsilon \gamma_0}-\frac{1}{\gamma_0},   
 P_1^{\rm I} =  \epsilon_1,
 P_2^{\rm I}=\epsilon_1 -\frac{1}{\gamma_2},
\end{align} 
\begin{align}
 P_1^{\rm II}=P_0^{\rm II}&= \sqrt{\frac{e^R}{ \gamma_0\gamma_2}} -\frac{1}{\gamma_0}, P_2^{\rm II} =\sqrt{\frac{e^R}{ \gamma_0\gamma_2}} -\frac{1}{\gamma_2},
\end{align}
and
\begin{align}
&P_1^{\rm III}=P_0^{\rm III}=\frac{\epsilon}{\gamma_1-\epsilon \gamma_0}
,  P_2^{\rm III}=\frac{e^R(\gamma_1-\epsilon \gamma_0)}{\gamma_1\gamma_2} -\frac{1}{\gamma_2}.
\end{align}
where $\epsilon_1=\sqrt{\frac{\epsilon e^R}{ \gamma_2\gamma_1}} $.
\end{lemma}
\begin{proof}
See Appendix \ref{proof4}. 
\end{proof}

Recall the fact that problem \eqref{pb:3} is convex, and hence its optimal solution has to satisfy  the Karush–Kuhn–Tucker (KKT) conditions \cite{Boyd}. Therefore,  
the optimal solution of problem \eqref{pb:3} can be obtained by using the KKT conditions as the certificate and selecting  the optimal solution from  the two pure NOMA solutions shown in Lemma \ref{lemma3} and the three H-NOMA solutions in Lemma \ref{lemma4}. 
 
{\it Remark 6:} As shown in Lemmas \ref{lemma3} and \ref{lemma4}, even for the simple two-user special case, it is challenging to obtain a concise  closed-form expression for the optimal solution of problem \eqref{pb:1}. Therefore, the next section is provided to show how to obtain the optimal solution of problem \eqref{pb:1} in an algorithmic manner.

 
\section{Multi-User Power Minimization for BackCom Assisted H-NOMA}\label{section 4}
In this section,  the general multi-user power minimization problem shown in \eqref{pb:1} is focused. In particular,    a low-complexity SCA based algorithm is developed first to yield  a sub-optimal solution of problem \eqref{pb:1}, followed by a BB-based algorithm which   leads to the optimal solution. 

\subsection{SCA-Based Resource Allocation }
Problem \eqref{pb:1} can be first recast in the following more compact form:
  \begin{problem}\label{pb:4} 
  \begin{alignat}{2}
\underset{P_{m}, \eta_{mi}\geq0  }{\rm{min}} &\quad   \sum^{M}_{m=1} P_m \label{4tst:1}
\\ s.t. &\quad    \sum^{m}_{i=1} \log\left(
1+\frac{\gamma_{mi}  {\eta_{mi}P_i} }{\sum^{M}_{j=m+1}\gamma_{ji}  {\eta_{ji}P_i} +1}
\right) \geq R,  \nonumber 
\\&\quad  1\leq m \leq M, \label{4tst:2} 
\\ &\quad  \eta_{mi}\leq 1,   1\leq i \leq m-1,  1\leq m \leq M,
  \end{alignat}
\end{problem} 
where $\gamma_{mi}=|h_m|^2|g_{mi}|^2$, except for the case of $\gamma_{mm}=|h_m|^2$.

The multiplications between the reflecting coefficients and power allocation parameters make  the problem complicated. Inspired by the discussions provided to the two-user special case, problem \eqref{pb:4} can be recast    as follows:
  \begin{problem}\label{pb:5} 
  \begin{alignat}{2}
\underset{P_{mi}\geq0  }{\rm{min}} &\quad   \sum^{M}_{m=1} P_{mm} \label{5tst:1}
\\ s.t. &\quad   \sum^{m}_{i=1}\log\left(
1+\frac{\gamma_{mi}  {P_{mi} } }{\sum^{M}_{j=m+1}\gamma_{ji}  {P_{ji}} +1}
\right) \geq R,  \nonumber \\&\quad  1\leq m \leq M, \label{5tst:2} 
\\ &\quad  P_{mi}\leq P_{ii},   1\leq i \leq m-1,  1\leq m \leq M, \label{5tst:3} 
  \end{alignat}
\end{problem} 
where $P_{mi} = \eta_{mi}P_i$, and $P_{mm}=P_m$. The new constraint \eqref{5tst:3} is required to ensure $\eta_{mi}\leq 1$. 
 
In \cite{hnomadown,9679390}, a low-complexity Pareto-optimal algorithm based successive resource allocation was developed for conventional H-NOMA. We note that such a successive resource allocation approach is no longer applicable to BackCom assisted H-NOMA, which is due to the fact that   problem \eqref{pb:5} is  more complicated   than those  conventional H-NOMA ones. For example, for conventional H-NOMA, constraint \eqref{5tst:3}  is not needed. As a result, it is possible to optimize one user's parameters  without knowing the other users' choices. However, with constraint \eqref{5tst:3}, ${\rm U}_m$'s transmit power during the $i$-th time slot is affected by ${\rm U}_i$'s choice of $P_{ii}$. 

It is important to point out that the non-convexity of problem \eqref{pb:5} is caused by   constraint \eqref{5tst:2} which is the difference between two concave functions. Recall that the difference between two concave functions can be approximated by applying SCA. To facilitate the implementation of SCA, problem \eqref{pb:5} can be first rewritten as follows: 
  \begin{problem}\label{pb:z1} 
  \begin{alignat}{2}
\underset{P_{mi}\geq0  }{\rm{min}} &\quad   \sum^{M}_{m=1} P_{mm} \label{z1tst:1}
\\ s.t. &\quad   \sum^{m}_{i=1}\log\left(1+
  \sum^{M}_{j=m}\gamma_{ji}  {P_{ji}}  
\right)-I_{mi} \geq R,    1\leq m \leq M, \label{z1tst:2} 
\\ &\quad  \eqref{5tst:3} ,
  \end{alignat}
\end{problem} 
where $
I_{mi} = \log\left(
  \sum^{M}_{j=m+1}\gamma_{ji}  {P_{ji}} +1
\right)$. 

The term $I_{mi}$ is a concave function, but  can be approximated as an affine function by using SCA. In particular, SCA can be implemented in an iterative manner, where during the $k$-th iteration, $I_{mi}$ can be approximated as the following affine function via the Taylor  approximation:
\begin{align}
I_{mi} =& \log\left(
  \sum^{M}_{j=m+1}\gamma_{ji}  {P_{ji}} +1
\right) = \log\left(
\mathbf{h}_{mi}^T\mathbf{p}+1
\right)\\\nonumber
\approx&\log\left(
\mathbf{h}_{mi}^T\mathbf{p}^{k-1}+1
\right)+ \frac{\mathbf{h}_{mi}^T(\mathbf{p}-\mathbf{p}^{k-1})}{\mathbf{h}_{mi}^T\mathbf{p}^{k-1}+1}
\end{align}
where $\mathbf{p}_m = \begin{bmatrix}
P_{m1}&\cdots&P_{mm}
\end{bmatrix}^T$, $\mathbf{p} = \begin{bmatrix}
\mathbf{p}_1^T &\cdots&\mathbf{p}_M^T 
\end{bmatrix}^T$, $\mathbf{p}^{k-1}$ is obtained from the previous iteration,    $\mathbf{h}_m = \begin{bmatrix}
\gamma_{m1}&\cdots&\gamma_{mm}
\end{bmatrix}^T$, $\mathbf{h} = \begin{bmatrix}
\mathbf{h}_1^T &\cdots&\mathbf{h}_M^T 
\end{bmatrix}^T$,  $\mathbf{h}_{mi} $ is structured similar to   $\mathbf{h}$ by including   $\gamma_{ji}$, $j\in \{m+1,\cdots,  M\}$, only.

\subsection{BB-Based Resource Allocation}
To characterize  the fundamental limits of BackCom assisted H-NOMA as well as the performance of the proposed SCA algorithm, it is important to develop an algorithm which yields the optimal solution of problem \eqref{pb:1}. 
A useful  observation from \eqref{pb:5} is that the objective function of the considered optimization problem is a  simple summation of $M$ optimization  parameters, which can be viewed as the vertex of  a multi-dimensional rectangle. This observation motivates the application of the BB algorithm which is to  solve the following optimization problem:
  \begin{problem}\label{pb:6} 
  \begin{alignat}{2}
\underset{  }{\rm{min}} &\quad   \sum^{M}_{m=1} P_{mm} \label{6tst:1}  &\quad s.t.  \quad  P_{mm}\in \mathcal{G},
  \end{alignat}
\end{problem}
where $\mathcal{G}=\{P_{mm}, 1\leq m \leq M: \eqref{5tst:2}, \eqref{5tst:3}\}$. 
 
 The BB algorithm can be viewed as a structured exhaustive search and needs to be implemented iteratively.  In brief, during the $i$-th iteration, an $M$-dimensional rectangle which is likely to contain the optimal solution, denoted by $\mathcal{R}_i$, is selected and split into two smaller rectangles, denoted by $\mathcal{R}_i^{1}$ and $\mathcal{R}_i^{2}$, respectively. The key step of each iteration is to find the upper and lower bounds on the optimal value of problem \eqref{pb:6} of the two rectangles, which are denoted by $\psi^{\rm U}\left(\mathcal{R}_i^{j}\right)$ and $\psi^{\rm L}\left(\mathcal{R}_i^{j}\right)$, $j\in \{\rm L,U\}$, respectively, and will be discussed later. At the end of each iteration, the global  upper and lower bounds, denoted by $U$ and $L$, respectively, are updated by using the existing rectangles,  i.e., $U=\min \psi^{\rm U}(\mathcal{R})$ and   $L=\min \psi^{\rm L}(\mathcal{R})$, $\forall \mathcal{R}\in \mathcal{C}$, where $\mathcal{C}$ denotes the set containing all possible rectangles. Those rectangles whose lower bounds are larger than the global upper bound need to be pruned. The details for the implementation of the BB algorithm are provided in Algorithm  \ref{algorithm}, where $N_{\max}$ denotes the maximal number of iterations.

 \begin{algorithm}[t]
\caption{Branch and Bound Algorithm}

 \begin{algorithmic}[1]
 
\State Initial rectangle set:   $\mathcal{C}$; find   $U=\psi^{\rm U}(\mathcal{C})$, $L=\psi^{\rm L}(\mathcal{C})$, $n=0$  
\If { $U-L\leq \xi$ or $n\leq N_{\rm max}$ }
\State $n=n+1$. 
\State  Find a rectangle $\mathcal{R}_i\in \mathcal{C} $ based on $\underset{i}{\arg}\min \psi^{\rm L}(\mathcal{R}_i)$ 
\State Split $\mathcal{R}_i$ along its longest edge into $\mathcal{R}^1_i$ and $\mathcal{R}^2_i$

\State Update  $\mathcal{C}$ by adding $\mathcal{R}^1_i$ and $\mathcal{R}^2_i$ and removing  $\mathcal{R}_i$

\State Update the global upper and lower bounds:   $U=\min \psi^{\rm U}(\mathcal{R})$ and   $L=\min \psi^{\rm L}(\mathcal{R})$, $\forall \mathcal{R}\in \mathcal{C} $.

\State Prune  the rectangles in $\mathcal{C}$ whose lower bounds are larger  than $U$.

 \EndIf
\State \textbf{end}
 \end{algorithmic}\label{algorithm}
\end{algorithm}  

As shown in Algorithm \ref{algorithm}, the most important component of a BB algorithm is   the functions, $\psi^{\rm U}(\mathcal{R})$ and $\psi^{\rm L}(\mathcal{R})$, i.e., given a rectangle, $\mathcal{R}$, how to find its  upper and lower bounds on the optimal value of the formulated optimization problem. An important observation is   that the objective function of the considered optimization problem is a monotonically increasing  function of $P_{mm}$, $1\leq m \leq M$, which motivates the use of   the  minimum and  maximum  vertices  of the rectangle, denoted by $\mathbf{p}_{\min}^{\mathcal{R}}$ and $\mathbf{p}_{\max}^{\mathcal{R}}$, respectively, for finding the two bounds, as follows:
\begin{align}
\left\{\begin{array}{l}  \psi^{\rm U}(\mathcal{R}) =\mathbf{1}_M^T \mathbf{p}_{\max}^{\mathcal{R}} \\
\psi^{\rm L}(\mathcal{R}) =\mathbf{1}_M^T \mathbf{p}_{\min}^{\mathcal{R}}
 \end{array}\right. ,
\end{align}
if $\mathbf{p}_{\max}^{\mathcal{R}}$ is feasible, i.e., $ \mathbf{p}_{\max}^{\mathcal{R}}\in \mathcal{G}$,  otherwise $\psi^{\rm U}(\mathcal{R}) =\psi^{\rm L}(\mathcal{R})=\Delta$,
where $\mathbf{1}_M$ is an $M\times 1$ all-one vector and $\Delta$ is a constant with a large value. Because OMA is a special case of H-NOMA, the overall power allocation of OMA can be used as $\Delta$.  

We note that checking the feasibility of $\mathbf{p}_{\max}^{\mathcal{R}}$ is equivalent to solve the following feasibility optimization problem:
  \begin{problem}\label{pb:7} 
  \begin{alignat}{2}
\underset{ \eta_{mi}\geq0  }{\rm{min}} &\quad  1
\\ s.t. &\quad   \sum^{m}_{i=1} \log\left(
1+\frac{\gamma_{mi}  {\eta_{mi}P_{ii}^{\max}} }{\sum^{M}_{j=m+1}\gamma_{ji}  {\eta_{ji}P_{ii}^{\max}} +1}
\right) \geq R, \\ \nonumber &\quad\quad\quad\quad  1\leq m \leq M, \label{1tst:7} 
\\ &\quad  \eta_{mi}\leq 1,   1\leq i \leq m-1,  1\leq m \leq M,
  \end{alignat}
\end{problem} 
where $P_{ii}^{\max}$ is the $i$-th element of $\mathbf{p}_{\max}^{\mathcal{R}}$.

Comparing problem \eqref{pb:7} to the original formulation in \eqref{pb:5}, the number of optimization variables is reduced since $P_{ii}^{\max}$ is fixed. However, problem \eqref{pb:7} is still a non-convex optimization problem, and hence challenging to solve.  In the following subsections, two approaches are introduced  to solve the feasibility problem in \eqref{pb:7}.

\subsubsection{An intuitive approach based SCA}\label{section scax}
The non-convexity feature of problem \eqref{pb:7} can be illustrated by recast it as follows: 
  \begin{problem}\label{pb:8} 
  \begin{alignat}{2}
\underset{P_{mi}\geq0  }{\rm{min}} &\quad   1 \label{1tst:8}
\\ s.t. &\quad   \sum^{m}_{i=1}\log\left(1+
  \sum^{M}_{j=m} \gamma_{ji} P_{ii}^{\max} {\eta_{ji}}
\right)-\tilde{I}_{mi} \geq R,  \nonumber \\ &\quad\quad\quad\quad  1\leq m \leq M, \label{8tst:2} 
\\ &\quad  \eta_{mi}\leq 1,   1\leq i \leq m-1,  1\leq m \leq M,
  \end{alignat}
\end{problem} 
where $
\tilde{I}_{mi} = \log\left(
  \sum^{M}_{j=m+1}\gamma_{ji} P_{ii}^{\max} {\eta_{ji}} +1
\right)$.  The left hand side of \eqref{8tst:2}  is the difference between two concave functions, which is not convex.  Hence, the overall optimization problem is not convex. 

An intuitive approach is to again apply SCA and convert the left-hand side    of \eqref{pb:8} to the difference between a concave function and an affine function. In particular, during the $k$-th iteration of SCA, the term $\tilde{I}_{mi}$ can be approximated as the following  affine function by using Taylor  approximation:
\begin{align}\nonumber 
\tilde{I}_{mi} =& \log\left(
  \sum^{M}_{j=m+1}\gamma_{ji} P_{ii}^{\max} {\eta_{ji}}  +1
\right) = \log\left(
{\boldsymbol \gamma}_{mi}^T{\boldsymbol \eta}+1
\right)\\ 
\approx&\log\left(
{\boldsymbol \gamma}_{mi}^T{\boldsymbol \eta}^{k-1}+1
\right)+ \frac{{\boldsymbol \gamma}_{mi}^T({\boldsymbol \eta}-{\boldsymbol \eta}^{k-1})}{{\boldsymbol \gamma}_{mi}^T{\boldsymbol \eta}^{k-1}+1},
\end{align}
where ${\boldsymbol \eta}_m = \begin{bmatrix}
\eta_{m1}&\cdots&P_{m(m-1)}
\end{bmatrix}^T$, ${\boldsymbol \eta}= \begin{bmatrix}
{\boldsymbol \eta}_1^T &\cdots&{\boldsymbol \eta}_M^T 
\end{bmatrix}^T$,  ${\boldsymbol \eta}^{k-1}$ is obtained from the previous iteration,   ${\boldsymbol \gamma}_m = \begin{bmatrix}
\gamma_{m1} P_1&\cdots&\gamma_{m(m-1)} P_{m-1}
\end{bmatrix}^T$, ${\boldsymbol \gamma} = \begin{bmatrix}
{\boldsymbol \gamma}_1^T &\cdots&{\boldsymbol \gamma}_M^T 
\end{bmatrix}^T$,  ${\boldsymbol \gamma}_{mi} $ is structured similar to   ${\boldsymbol \gamma}$ by  including     $\gamma_{ji}P_{ii}^{\max}$, $j\in \{m+1, M\}$, only.

While the use of  SCA   is intuitive, our simulation results indicate that the SCA approach  leads to high computational complexity and also  poor performance, which motivates the following alternative approach. 

\subsubsection{Successive resource allocation}\label{section sra} The feature of H-NOMA can be used to design a more efficient feasibility algorithm, as shown in the following. In particular, recall that in each time slot, the signals from ${\rm U}_M$ are decoded during the last SIC stage, which means that ${\rm U}_M$'s signals do not suffer interference from the other users. This observation can be illustrated by excluding ${\rm U}_m$'s constraint, $1\leq m \leq M-1$ and rewriting problem \eqref{pb:7} as follows:
  \begin{problem}\label{pb:9} 
  \begin{alignat}{2}
\underset{ \eta_{Mi}\geq0  }{\rm{min}} &\quad   1  \label{9tst:1}
\\ s.t. &\quad   \sum^{M}_{i=1}  \log\left(
1+ |h_M|^2|g_{Mi}|^2  {\eta_{Mi}P_{ii}^{\max}} 
\right)\geq R    \label{9tst:2} 
\\ &\quad \sum^{M}_{i=1}\eta_{Mi}\leq \eta.
  \end{alignat}
\end{problem} 
As shown in \eqref{pb:9}, ${\rm U}_M$'s parameters can be optimized without knowing the other users' choices of $\eta_{m,i}$, $1\leq m \leq M-1$. In addition, problem \eqref{pb:9} is   a convex optimization problem, and can be efficiently solved by those off-shelf optimization solvers. 

Once ${\rm U}_M$'s parameters are optimized, ${\rm U}_{M-1}$'s can then be optimized, by treating $\eta_{Mi}$ as fixed. Similarly, all the users' reflecting coefficients can be obtained  efficiently in a successive manner, i.e., during the $m$-th stage, the following convex optimization problem needs to be solved:
 \begin{problem}\label{pb:10} 
  \begin{alignat}{2}
\underset{ \eta_{mi}\geq0  }{\rm{min}} &\quad  1
\\ s.t. &\quad   \sum^{m}_{i=1} \log\left(
1+\frac{\gamma_{mi}  {\eta_{mi}P_{ii}^{\max}} }{\sum^{M}_{j=m+1}\gamma_{ji}  {\eta_{ji}P_{ii}^{\max}} +1}
\right) \geq R,   \label{1tst:10} 
\\ &\quad  \eta_{mi}\leq 1,   1\leq i \leq m-1,   
  \end{alignat}
\end{problem} 
by assuming that the parameters of ${\rm U}_j$, $m+1\leq j\leq M$, are fixed. 

{\it Remark 7:} Compared to the SCA based feasibility approach, the computational complexity of  the proposed  one  is smaller, due to  the following two reasons. One reason  is that the successive resource allocation based  approach can be implemented without iterations as required by SCA. The other is that   each subproblem shown in \eqref{pb:10} is convex and the number of its optimization variables is simply $m-1$, whereas the dimension of the optimization vector for the SCA problem in \eqref{pb:8} is $N_{\rm SCA} = \frac{(M-1)M}{2}$. Therefore, the use of  the interior point method to solve the SCA based feasibility problem leads to the computational complexity proportional to $\mathcal{O}( N_{\rm SCA}^{3.5})$  \cite{scacomplexity}. If  successive resource allocation is used, the computational complexity for solving the feasibility problem is reduced to   $\sum^{M}_{m=1}\mathcal{O}( m^{3.5})$, which can be much smaller that of SCA, particularly for large $M$.  

\section{Numerical Studies} \label{section 5}
 In this section, the performance of BackCom assisted H-NOMA is studied by using computer simulation results. For all the conducted simulations, the $M$ users are assumed to be randomly clustered within a square with its side length denoted by  $r_u$ and its center located at  $(r_c, r_c)$ m, where the base station is located at the origin.   The channels between the users and the base station are subject to both Rayleigh fading  and large scale path loss. The channels between the users are subject to both Rician fading and large scale path loss. The maximal number of iterations for the implementation of the BB algorithm is set as $N_{\rm max}=1000$. 

In Figs. \ref{fig1} and \ref{fig2}, the two-user special case discussed in Section  \ref{section 3} is focused on. We note that for this special case,  the optimal solution for the considered power minimization problem can be obtained. In particular, for the two figures, the analytical results   are obtained by combining Lemmas \ref{lemma3} and \ref{lemma4}, where the simulation results  are obtained by applying the off-shelf optimization solvers. As can be seen from   both the   figures, the 
analytical results perfectly match the simulation results, which verifies the accuracy of the developed analytical results. 

In both  Figs. \ref{fig1} and \ref{fig2}, the performance of OMA is also shown as a benchmarking scheme. As can be seen from the figures, the use of H-NOMA can significantly reduces the energy consumption of OMA, particularly if the target data rate is large. Comparing Fig. \ref{fig1a} to Fig. \ref{fig1b}, one can observe that the case with a smaller choice of $r_u$ yields a large performance gap between OMA and H-NOMA, which can be explained in the following. For the case with  smaller $r_u$, the users are clustered in a smaller region, which makes the inter-user channels stronger, i.e., the cooperation among the users becomes more meaningful.  Comparing Fig. \ref{fig1} to Fig. \ref{fig2}, one can observe that the less noisy the environment is,  the less power consumption is required, where H-NOMA always outperforms OMA regardless of the choices of the system parameters.

     \begin{figure}[t] \vspace{-0em}
\begin{center}
\subfigure[$r_u = 2$ m ]{\label{fig1a}\includegraphics[width=0.38\textwidth]{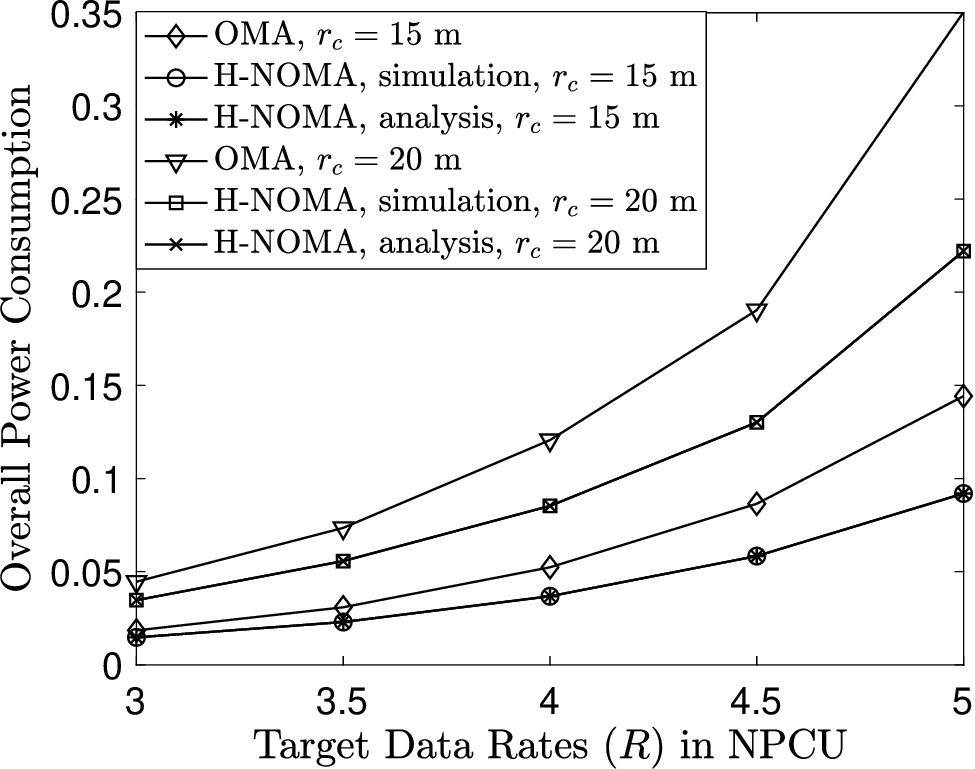}} 
\subfigure[$r_u = 5$ m]{\label{fig1b}\includegraphics[width=0.38\textwidth]{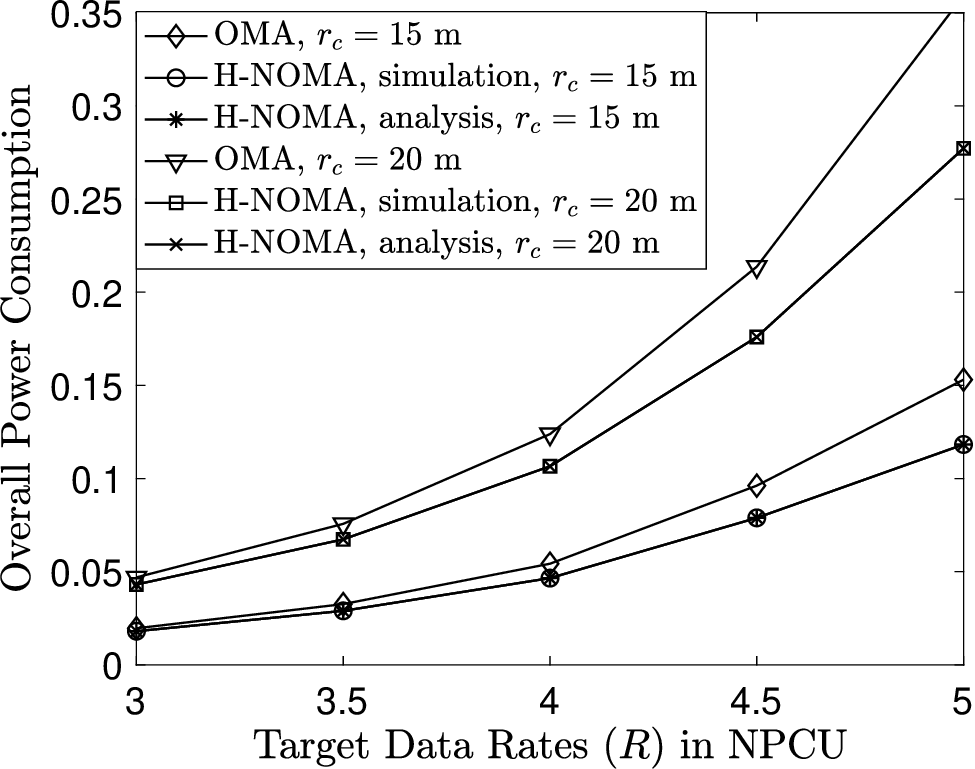}} \vspace{-1em}
\end{center}
\caption{Overall power consumption required by  the considered transmission schemes, where $M=2$,   the noise power is set as $\sigma^2=10^{-8}$, and NPCU denotes nats per channel use.        \vspace{-1em} }\label{fig1} 
\end{figure}
     \begin{figure}[t] \vspace{-0em}
\begin{center}
\subfigure[$r_u = 2$ m ]{\label{fig2a}\includegraphics[width=0.38\textwidth]{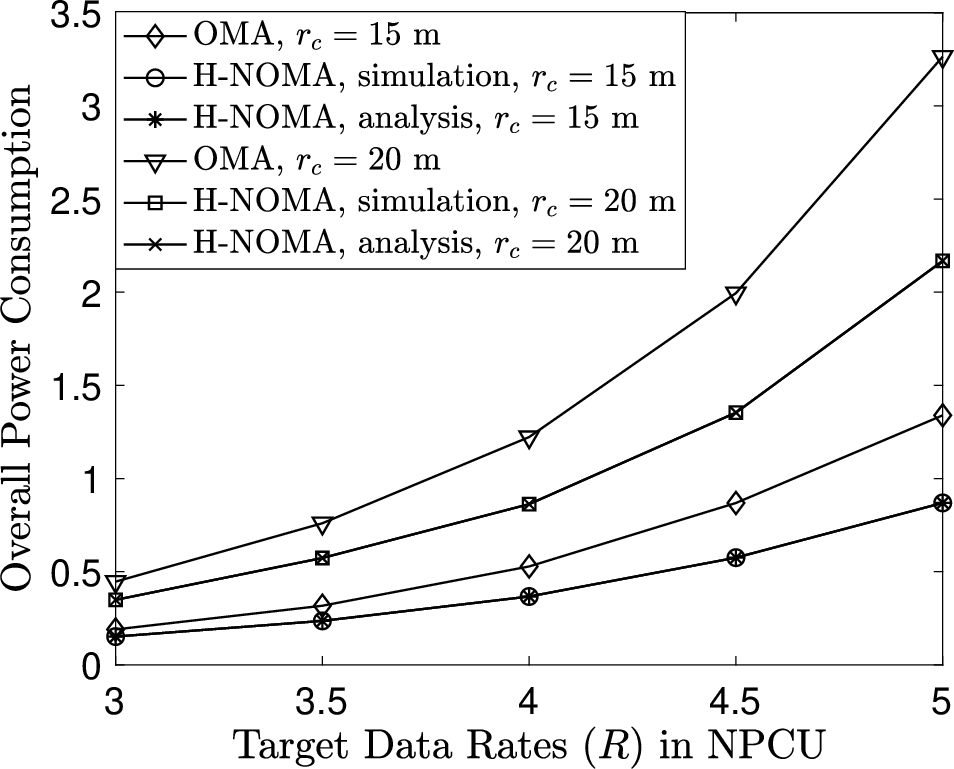}} 
\subfigure[$r_u = 5$ m]{\label{fig2b}\includegraphics[width=0.38\textwidth]{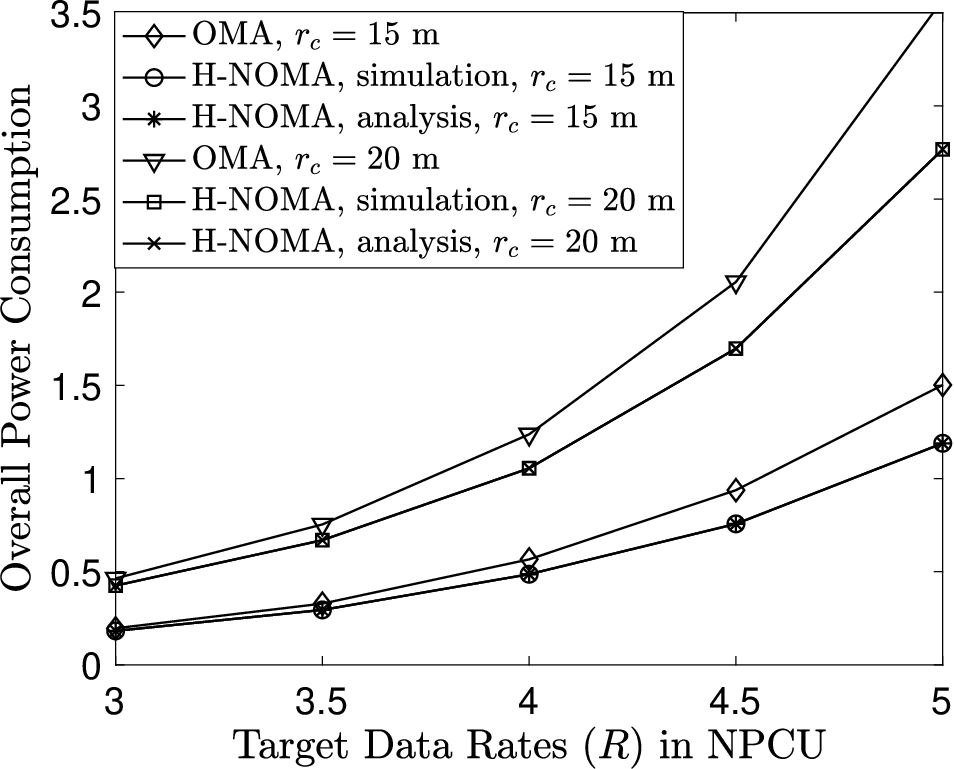}} \vspace{-1em}
\end{center}
\caption{Overall power consumption required by  the considered transmission schemes, where $M=2$,   and the noise power is set as $\sigma^2=10^{-7}$.       \vspace{-1em} }\label{fig2} 
\end{figure}

As discussed in Section \ref{section 3}, it is more challenging to carry out   resource allocation for BackCom assisted H-NOMA than conventional H-NOMA. In particular, Lemmas \ref{lemma3} and \ref{lemma4} show that there are potentially  five forms of the optimal solution of problem \eqref{pb:3}, including the two pure NOMA solutions shown in Lemma \ref{lemma3} and the three H-NOMA solutions shown in Lemma \ref{lemma4}. Fig. \ref{fig3} shows the probability for these potential forms to be used as the optimal solution. For the case with small $R$, the dominant solution is H-NOMA Type III, i.e., $P_i^{\rm III}$ shown in Lemma \ref{lemma4}. For the case with large $R$, the dominant solution is H-NOMA Type I, i.e., $P_i^{\rm I}$ shown in Lemma \ref{lemma4}. Recall that the key difference between the two types of H-NOMA solutions is whether $P_0=P_1$, i.e., whether ${\rm U}_2$'s reflecting coefficient is $one$. Therefore, an important observation from Fig. \ref{fig3} is that a user does not have to reflect in a `full blasting' manner if $R$ is small. Another interesting observation from Fig. \ref{fig3} is that the pure NOMA mode is possible to be used if $R$ is small, and almost never be used for the case with large $R$.

     \begin{figure}[t]\centering \vspace{-0.2em}
    \epsfig{file=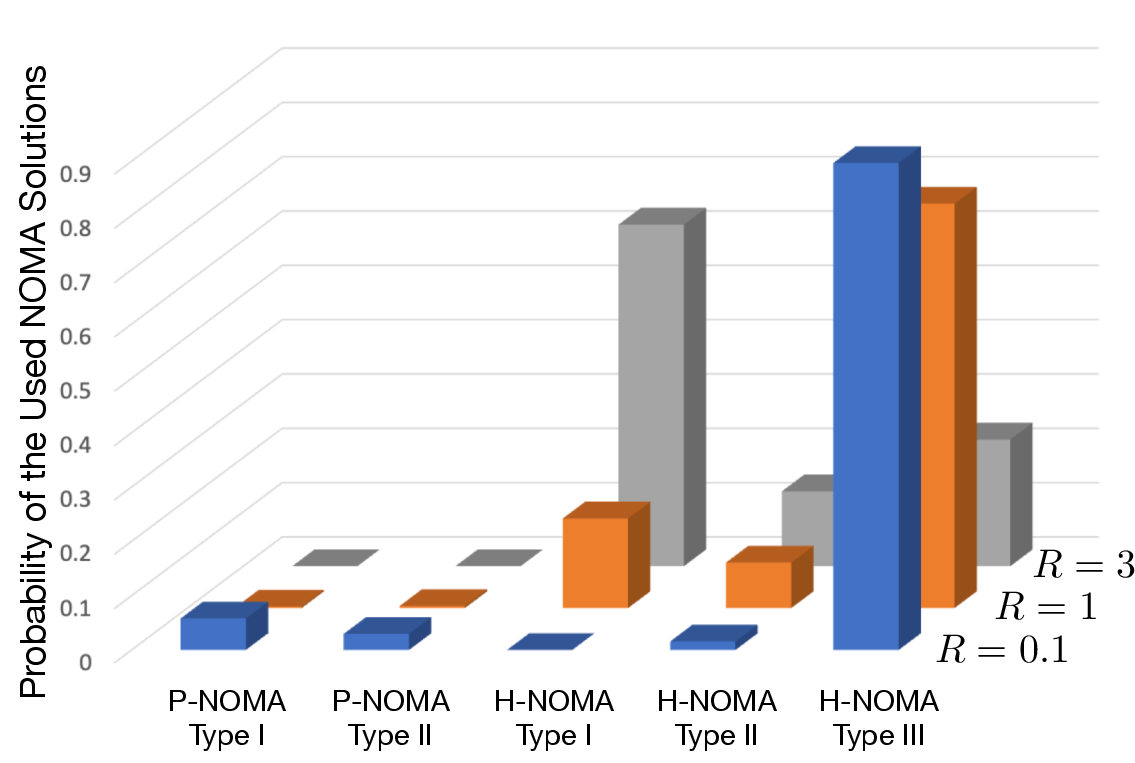, width=0.4\textwidth, clip=}\vspace{-0.5em}
\caption{Probability of the five NOMA solutions to become the optimal solution, where the two pure NOMA solutions from Lemma \ref{lemma3}, denoted by P-NOMA Types I and II, respectively, and the three H-NOMA solutions from Lemma \ref{lemma4}, denoted by H-NOMA Types I  II and III, respectively,   are used.   $M=2$, $\sigma^2=10^{-8}$, $r_u=2$ m,  and $r_c=15$ m. 
  \vspace{-1em}    }\label{fig3}   \vspace{-0.5em} 
\end{figure}

     \begin{figure}[t] \vspace{-0em}
\begin{center}
\subfigure[$r_c = 15$ m ]{\label{fig4a}\includegraphics[width=0.38\textwidth]{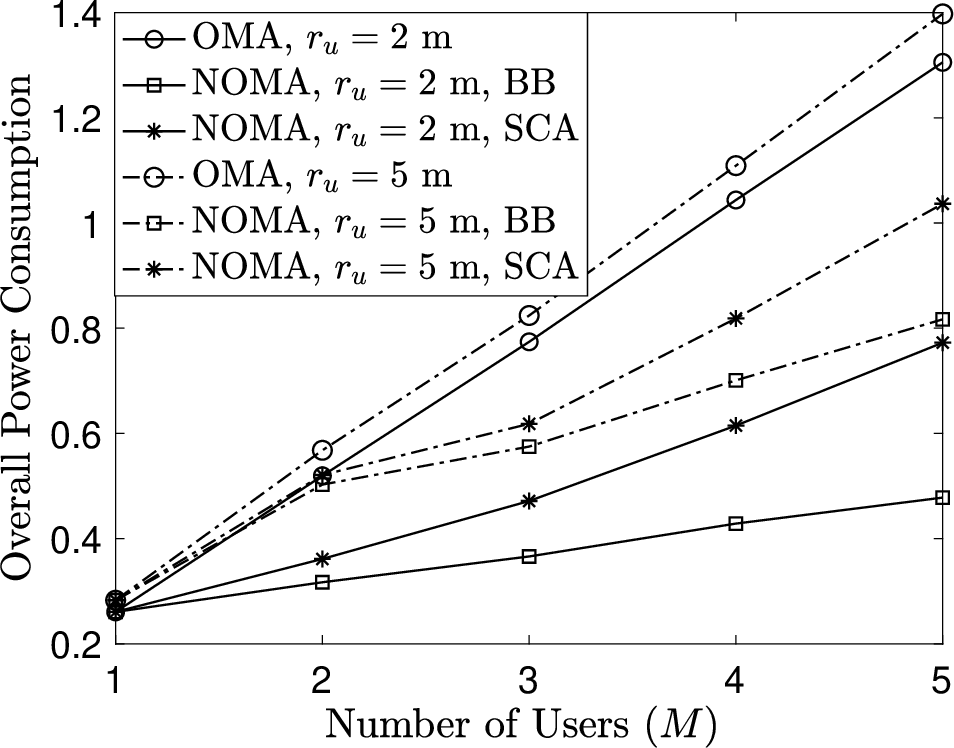}} 
\subfigure[$r_c = 20$ m]{\label{fig4b}\includegraphics[width=0.38\textwidth]{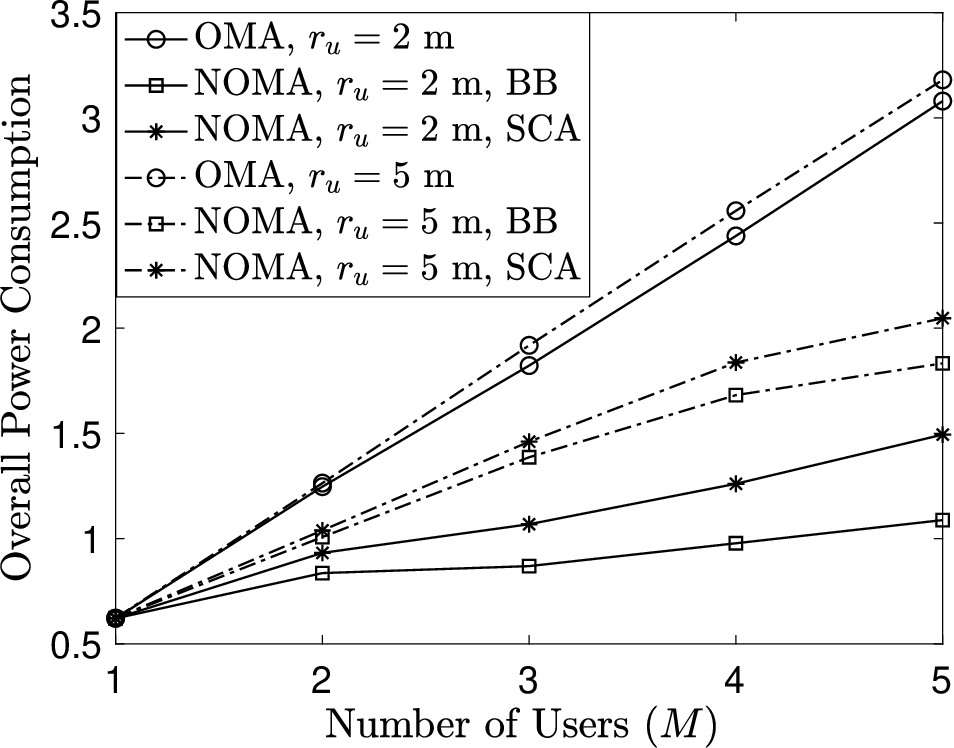}} \vspace{-1em}
\end{center}
\caption{Overall power consumption required by  the considered transmission schemes in the multi-user scenario, where   $R=4$ NPCU and the noise power is set as $\sigma^2=10^{-8}$.      \vspace{-1em} }\label{fig4} 
\end{figure}

In Fig. \ref{fig4}, the performance of BackCom assisted H-NOMA is evaluated in the general multi-user scenario, where the  successive resource allocation based BB algorithm shown in Section \ref{section sra} is used. As can be seen from  the two figures in Fig. \ref{fig4}, the H-NOMA schemes always outperform the OMA benchmarking scheme. An important observation from Fig. \ref{fig4} is that the gap between OMA and H-NOMA increases with the growth of $M$, which can be explained as follows. For the case of $M=1$, there is a single user in the cluster, and hence H-NOMA is degraded to OMA. With more users participating  into cooperation, there are   larger degrees of freedom for resource allocation. For example, for the case of $M=2$, ${\rm U}_M$ has the choice to transmit in two time slots only, whereas for the case of $M=5$, ${\rm U}_M$ has the choice to transmit in five time slots. Even though  ${\rm U}_M$ might not choose to transmit in all the five time slots, there is more flexibility for resource allocation, which leads to a large performance gain  over OMA.

Fig. \ref{fig4} also illustrates that the performance gap between OMA and H-NOMA becomes larger if $r_u$ is smaller, which is consistent to the observation from Figs. \ref{fig1} and \ref{fig2}. It is interesting to observe that a change of $r_u$ does not bring much performance change to OMA, since the OMA performance is mainly affected by $r_c$, i.e., the distance between the base station and the center of the cluster. Recall that the motivation to develop the SCA algorithm is that it achieves a balanced tradeoff between the system performance and complexity. This is the reason why Fig. \ref{fig4}   shows that the use of SCA might cause a performance loss compared to the BB algorithm. It is worth to point out that with a smaller number of users,  SCA can realize a performance close to the optimal performance; however, an increase of $M$ introduces a larger performance gap between the BB and SCA algorithms.   Fig. \ref{fig5} shows the overall power consumption as a function of the target data rate. As can be observed from the figure, by increasing the target data rate, the performance gap between OMA and H-NOMA is increased. For example, with $R=5$ nats per channel use (NPCU), the use of OMA leads to the power consumption which is almost  three times that of H-NOMA.

     \begin{figure}[t]\centering \vspace{-0.2em}
    \epsfig{file=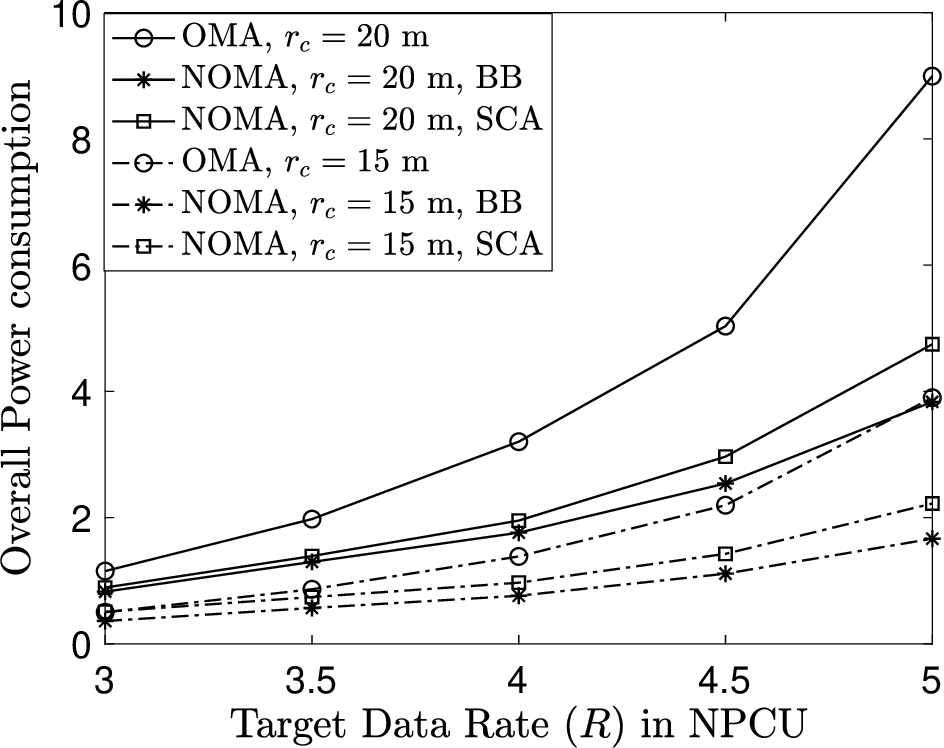, width=0.38\textwidth, clip=}\vspace{-0.5em}
\caption{Overall power consumption required by  the considered transmission schemes as a function of the target data rate, where   $M=5$, the noise power is set as $\sigma^2=10^{-8}$, and $r_u=5$ m.
  \vspace{-1em}    }\label{fig5}   \vspace{-0.5em} 
\end{figure}

     \begin{figure}[t]\centering \vspace{-0.2em}
    \epsfig{file=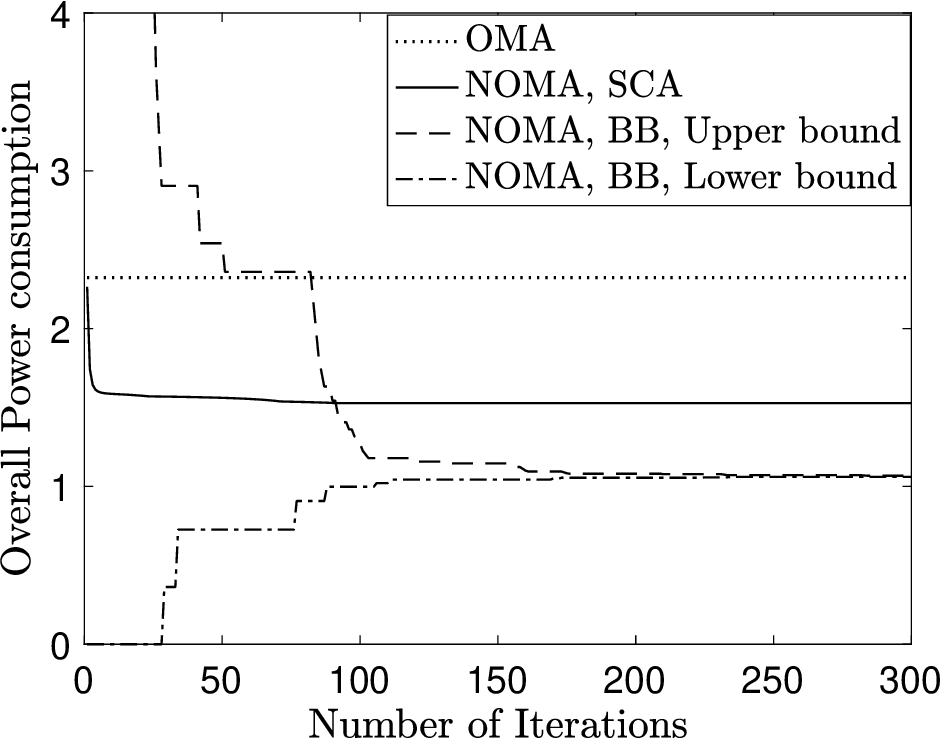, width=0.38\textwidth, clip=}\vspace{-0.5em}
\caption{Convergence of the SCA and BB algorithms, where   $M=5$, the noise power is set as $\sigma^2=10^{-8}$, $R=4$ NPCU,  $r_u=5$ m, and $r_c=20$ m.
  \vspace{-1em}    }\label{fig6}   \vspace{-0.5em} 
\end{figure}

\begin{table}[t]\vspace{-0em}
  \centering
  \caption{Running Time Required by   SCA and SRA Based BB Algorithms (in seconds)  }
  \begin{tabular}{|c|c|c|c|c|c|}
\hline
&    $M=1$& $M=2$  & $M=3$ &$M=4$ &$M=5$   \\
    \hline
    SCA  &$   0.2393$ 	&$5.8233$  	&$45.3231$  	&$100.3418$  	&$110.9291$  
    \\
    \hline
    SRA &$ 0.3271$  	&$49.8945$  	&$555.9374$  &	$698.0526$  &	$630.5155$   \\\hline
  \end{tabular}\vspace{1em}\label{table2}\vspace{-1em}
\end{table}  

The convergence of the proposed SCA and BB algorithms is studied in Fig. \ref{fig6}. As can be seen from the figure, the BB algorithm   achieves better performance than SCA, but requires a significantly large number of iterations to converge, i.e., the computational complexity of SCA is much smaller than that of the BB algorithm.  As discussed in Section \ref{section scax}, instead of the successive resource allocation based feasibility approach,  SCA can also be an alternative for solving the feasibility problem of the BB algorithm. Table \ref{table2} shows the running time required by the SCA and successive resource allocation (SRA) based BB algorithms, where the choices of the system parameters are the same as those in Fig. \ref{fig6}. As can be seen from the table, the use of SCA to solve the BB feasibility problem causes significantly higher computational complexity than that of SRA. For example, with large $M$, the use of SCA requires around $6$ times more running time than SRA, an observation consistent to Remark 7. The use of SCA to solve the BB feasibility problem does not only cause more computational complexity, but also introduce performance losses, as shown in Table \ref{table2}\footnote{ It is worth to point out that Table \ref{table2} is obtained with $200$ iterations only, due to the extremely large running time of the SCA-based BB algorithm. }. In particular, the table shows that  the SCA-based BB algorithm cannot even outperform SCA. This observation might be due to the fact that the maximal number of iterations for the BB algorithm is set as $1000$, and   the SCA-based BB algorithm may not always converge before $1000$ iterations.  

\begin{table}[t]\vspace{-0em}
  \centering
  \caption{Performance of    the   SCA Based        BB Algorithm   }
  \begin{tabular}{|c|c|c|c|c|c|}
\hline
&    $M=1$& $M=2$  & $M=3$ &$M=4$ &$M=5$   \\
    \hline
    OMA  &$0.5837$&  $	1.4140$  &	$1.5670$  &	$2.4036$  &	$3.2109$      \\
    \hline
    SCA & $0.5837$  &$	1.1920$  &	$1.1594$  	&$1.7114$  &$	2.0892$    \\
    \hline
    SCA-based BB  &$0.5872$  &$ 1.1952$  	&$1.2270$  &	$1.7479 $  &	$2.5876$   
     \\\hline
  \end{tabular}\vspace{1em}\label{table3}\vspace{-1em}
\end{table}  

\section{Conclusions}\label{section 6}
In this paper, a new BackCom assisted H-NOMA uplink transmission scheme which is  a general multiple access framework of conventional OMA and pure NOMA has been developed.   Resource allocation for H-NOMA uplink transmission has also been  considered, where an overall power minimization problem was formulated. Insightful understandings for the key features of BackCom assisted H-NOMA and its difference from conventional H-NOMA have been illustrated by developing analytical results for the two-user special case. For the general multi-user scenario, two algorithms, one based on the BB principle and the other based on SCA, have been developed to realize different tradeoffs between the system performance and complexity. The numerical results have also been   provided to verify the accuracy of the developed analytical results and demonstrate the performance gain of H-NOMA over OMA. An important direction for future research is to use other performance metrics, such as energy efficiency, for  resource allocation  in BackCom assisted H-NOMA networks. 
 
\appendices 
\section{Proof of Lemma \ref{lemma1}}\label{proof1}
By explicitly showing the constraints of $P_i\geq 0$, $i\in\{0,1,2\}$, problem \eqref{pb:3} can be first rewritten 
 as follows:
\begin{problem}\label{pb:x0} 
  \begin{alignat}{2}
\underset{ }{\rm{min}} &\quad   P_1+P_2 \label{0tst:x0}
\\   s.t. &\quad   \eqref{3tst:1},   \eqref{3tst:2}, \eqref{3tst:3},  P_i\geq 0, i\in\{0,1,2\},
  \end{alignat}
\end{problem} 
which is a convex optimization problem. Its   Lagrangian  is given by 
\begin{align}\nonumber
L= &P_1+P_2 +\lambda_1\left(
 R- \log\left(
1+\gamma_0  P_0 
\right)-\log\left(
1+ \gamma_2P_2
\right)
 \right)\\ &+\lambda_2\left(
 \gamma_0 P_0+\epsilon-\gamma_1 P_1
 \right)+\lambda_3(P_0-P_1)\\\nonumber &-\lambda_4  P_1-\lambda_5  P_2-\lambda_6  P_0 ,
\end{align}
where $\lambda_i$'s, $1\leq i \leq 6$, denote the Lagrange  multipliers. 
Recall that for a convex optimization problem, a solution is optimal if and only if it satisfies the KKT conditions.  Therefore, in order to show that the OMA solution is not  optimal to problem \eqref{pb:x0}, it is sufficient to show that the OMA solution  cannot satisfy the corresponding  KKT conditions. It is straightforward to show that the OMA solution satisfies all the primal feasibility conditions. The stationarity and complementary slackness conditions of  
 \eqref{pb:x0} are given by 
\begin{align}
\left\{\begin{array}{l} 1-\gamma_1\lambda_2-\lambda_3-\lambda_4=0 ,
1-\frac{\lambda_1 \gamma_2}{1+\gamma_2P_2}-\lambda_5=0,\\  
-\frac{\lambda_1 \gamma_0}{1+\gamma_0P_0}+\epsilon \lambda_2 \gamma_0+\lambda_3-\lambda_6=0,
\\  
\lambda_1\left(R- \log\left(
1+\gamma_0  P_0 
\right)-\log\left(
1+ \gamma_2P_2
\right)\right)=0,
\\    
 \lambda_2\left( \epsilon \gamma_0 P_0+\epsilon-\gamma_1 P_1\right)= 0, \lambda_3\left(  P_0- P_1\right)=0,
 \\    \lambda_4P_1=0,\lambda_5P_2=0,\lambda_6P_0=0.
 \end{array}\right.
\end{align}
By using the fact that  for the OMA solution,  $ P_m=\frac{\epsilon}{\gamma_m}$, $m=\{1,2\}$, and  $P_0=0$, with some straightforward algebraic manipulations, the closed-form expressions of the Lagrange multipliers can be obtained   as follows:
\begin{align}
& \lambda_1   =\frac{1+  \epsilon}{\gamma_2}, \lambda_2 =\frac{1}{\gamma_1}, \lambda_3   =0, \lambda_4 =0,\lambda_5 =0  \\\nonumber& 
\lambda_6= -\lambda_1 \gamma_0 +\epsilon \lambda_2 \gamma_0 =\frac{\epsilon   \gamma_0 }{\gamma_1}-  \frac{\gamma_0+\gamma_0  \epsilon}{\gamma_2}  .
\end{align}
If the OMA solution is not optimal, at least one of these Lagrange multipliers should be negative. In particular,  $\lambda_6$ can be expressed as follows:
\begin{align}
\lambda_6 &=\gamma_0  \frac{\epsilon     \gamma_2 -  \gamma_1-   \gamma_1\epsilon}{\gamma_1\gamma_2}     
    \\\nonumber &  =  |g_{21}|^2 \frac{\epsilon \left(   |h_2|^2- |h_1|^2\right) -   |h_1|^2}{|h_1|^2 }  < 0,
\end{align}
where     the fact that $ |h_2|^2< |h_1|^2$ is used. Therefore,   the OMA solution cannot be an optimal solution, and the proof of the lemma is complete. 

\section{Proof of Lemma \ref{lemma2}} \label{proof2}
Recall that the overall power minimization problem for conventional H-NOMA can be expressed as follows: 
 \begin{problem}\label{pb:x3} 
  \begin{alignat}{2}
\underset{ P_i\geq0 }{\rm{min}} &\quad   P_1+P_{0}+P_{2} \label{1tst:x3}
\\ s.t.  \label{2tst:x3}&\quad      \log\left(
1+|h_2|^2 P_{0}
\right)+ \log\left(
1+ |h_2|^2P_{2}
\right)\geq R
\\&\quad  \label{3tst:x3}
   \log\left(
1+\frac{|h_1|^2 P_1 }{ |h_2|^2 P_{0}+1}
\right)\geq R,
  \end{alignat}
\end{problem} 
where $P_0$ is shown in the objective function since in conventional H-NOMA, $P_0$ is ${\rm U}_2$'s transmit power during the first time slot, and    the constraint $P_0\leq P_1$ in problem \eqref{pb:3}  is no longer  needed for conventional H-NOMA.

The proof can be completed by first obtaining  the closed-form expression of the pure NOMA solution and then proving that it cannot be an optimal solution for conventional H-NOMA. 
 
In particular, for conventional H-NOMA, the pure NOMA solution can be obtained by using the fact that $P_2=0$ for pure NOMA and simplifying problem \eqref{pb:x3} to the following form: 
  \begin{problem}\label{pb:x2} 
  \begin{alignat}{2}
\underset{ P_i\geq0 }{\rm{min}} &\quad   P_1+P_{0}  \label{1tst:x2}
\\ s.t. &\quad      \log\left(
1+|h_2|^2 P_{0}
\right) \geq R
\\&\quad 
   \log\left(
1+\frac{|h_1|^2 P_1 }{ |h_2|^2 P_{0}+1}
\right)\geq R.
  \end{alignat}
\end{problem}
  With some straightforward algebraic manipulations,  the pure NOMA solution can be obtained as follows:  
\begin{align}\label{pure noma}
  P_1  =&  \frac{(\epsilon+1)\epsilon}{  |h_1|^2}, \quad 
   P_{0}=  \frac{ \epsilon   }{ |h_2|^2 }  .
\end{align}

Recall that problem \eqref{pb:x3} is convex, and hence  the proof can be complete by showing that the pure NOMA solution shown in \eqref{pure noma} does not satisfy the KKT conditions of problem \eqref{pb:x3}. Note that the Lagrangian of problem \eqref{pb:x3} is given by 
\begin{align}
L=& P_1+P_{0}+P_{2} +\lambda_1 (R-\log\left(
1+|h_2|^2 P_{0}
\right)\\\nonumber &- \log\left(
1+ |h_2|^2P_{2}
\right) ) +\lambda_2\left(\epsilon |h_2|^2 P_{0}+\epsilon -   |h_1|^2 P_1  
\right)\\\nonumber & -\lambda_3 P_1-\lambda_4 P_{0}-\lambda_5P_{2},
\end{align}
where   $\lambda_i$'s, $1\leq i \leq 5$, denote the Lagrange  multipliers. 
The use of the complementary slackness and stationarity conditions lead to the following observations: 
\begin{align}
\left\{\begin{array}{l} 
1-\lambda_2 |h_1|^2 -\lambda_3=0\\
1 - \frac{\lambda_1|h_2|^2}{1+|h_2|^2 P_{0}}+\lambda_2  \epsilon |h_2|^2-\lambda_4=0\\
1 -\frac{\lambda_1 |h_2|^2}{1+ |h_2|^2P_{2}}-\lambda_5=0\\
\lambda_1\left(R-\log\left(
1+|h_2|^2 P_{0}
\right)+ \log\left(
1+ |h_2|^2P_{2}
\right)\right)=0\\
\lambda_2\left(\epsilon |h_2|^2 P_{0}+\epsilon -   |h_1|^2 P_1  
\right)=0\\
\lambda_3 P_1=0,  \lambda_4 P_{0}=0,  \lambda_5P_{2}=0
 \end{array}\right..
\end{align}
By using the fact that for pure NOMA,  $P_1  =  \frac{(\epsilon+1)\epsilon}{  |h_1|^2}$,  $   P_{0}=  \frac{ \epsilon   }{ |h_2|^2 }  $ and $P_{2}=0$, the Lagrange multipliers can be expressed in the following expressions: 
\begin{align}
\left\{\begin{array}{l} 
 \lambda_2   =\frac{1}{|h_1|^2}\\
\lambda_1=\frac{1+  \epsilon    }{|h_2|^2}+  \frac{1+  \epsilon    }{|h_2|^2} \epsilon \frac{|h_2|^2}{|h_1|^2}  \\ 
\lambda_3  =0,  \lambda_4  =0,  \lambda_5=1 - \lambda_1 |h_2|^2 
 \end{array}\right..
\end{align}
Therefore,  whether the pure NOMA solution is optimal depends on   the value of $\lambda_5$ which can be expressed as follows:
\begin{align}
  \lambda_5=&1 -  |h_2|^2 \left(\frac{1+  \epsilon    }{|h_2|^2}+  \frac{1+  \epsilon    }{|h_2|^2} \epsilon \frac{|h_2|^2}{|h_1|^2}
  \right)\\\nonumber
  =& -  \epsilon   -   (1+  \epsilon ) \epsilon \frac{|h_2|^2}{|h_1|^2}< 0.
\end{align}
Therefore, for conventional H-NOMA, pure NOMA cannot be the optimal transmission strategy, and the proof of the lemma is complete. 

\section{Proof of Lemma \ref{lemma3}}\label{proof3} 
Assume  that pure NOMA is used, i.e.,   $P_2=0$, and problem \eqref{pb:3} can be simplified as follows: 
 \begin{problem}\label{pb:0} 
  \begin{alignat}{2}
\underset{P_{i}\geq0  }{\rm{min}} &\quad   P_1  \label{0tst:0}
\\  \label{0tst:1}s.t. &\quad       R- \log\left(
1+\gamma_0  P_0 
\right) \leq 0
\\&\quad  \label{0tst:2}
  \epsilon \gamma_0 P_0+\epsilon-\gamma_1 P_1\leq 0
, \quad P_0\leq P_1, 
  \end{alignat}
\end{problem} 
which can be further simplified as follows:
 \begin{problem}\label{pb:0} 
  \begin{alignat}{2}
\underset{P_{i}\geq0  }{\rm{min}} &\quad   P_1  \label{0tst:0}
\\  \label{0tst:1}s.t. &\quad        
   \frac{\epsilon  }{\gamma_0}\leq P_0 
 \leq P_1, \quad   
  P_1 \geq  \frac{\epsilon \gamma_0 P_0+\epsilon}{\gamma_1}. 
  \end{alignat}
\end{problem} 
With some straightforward algebraic manipulations, the pure NOMA solution can be expressed as follows:
\begin{align}
P_1 = \max\left\{\frac{\epsilon}{\gamma_0},\frac{\epsilon(1+\epsilon)}{\gamma_1}
\right\}, \quad
 P_0 = \frac{\epsilon}{\gamma_0}. 
\end{align}
%

Recall that the use of the complementary slackness and stationarity conditions of problem \eqref{pb:3} lead to the following conclusions: 
\begin{align}\label{kktcccc}
\left\{\begin{array}{l} 1-\gamma_1\lambda_2-\lambda_3-\lambda_4=0 ,
1-\frac{\lambda_1 \gamma_2}{1+\gamma_2P_2}-\lambda_5=0,\\  
-\frac{\lambda_1 \gamma_0}{1+\gamma_0P_0}+\epsilon \lambda_2 \gamma_0+\lambda_3-\lambda_6=0,
\\  
\lambda_1\left(R- \log\left(
1+\gamma_0  P_0 
\right)-\log\left(
1+ \gamma_2P_2
\right)\right)=0,
\\    
 \lambda_2\left( \epsilon \gamma_0 P_0+\epsilon-\gamma_1 P_1\right)= 0, \lambda_3\left(  P_0- P_1\right)=0,
 \\    \lambda_4P_1=0,\lambda_5P_2=0,\lambda_6P_0=0.
 \end{array}\right.
\end{align}

First, the case that   $P_0=\frac{\epsilon}{\gamma_0}$ and $P_1=\frac{\epsilon(1+\epsilon)}{\gamma_1}$ is focused on, which requires the following assumption:
\begin{align}
\frac{\epsilon}{\gamma_0}\leq \frac{\epsilon(1+\epsilon)}{\gamma_1} \Rightarrow \frac{ \gamma_1}{\gamma_0}\leq  (1+\epsilon). 
\end{align}
It is straightforward to verify that $P_0=\frac{\epsilon}{\gamma_0}$ and $P_1=\frac{\epsilon(1+\epsilon)}{\gamma_1}$ ensure that $R- \log\left(
1+\gamma_0  P_0 
\right)-\log\left(
1+ \gamma_2P_2
\right)=0$ and $\epsilon \gamma_0 P_0+\epsilon-\gamma_1 P_1=0$, which implies  that $\lambda_1\neq 0$ and  $\lambda_2\neq 0$. 
 By using the facts that $P_0\neq P_1\neq 0$ and  $P_2=0$ for the pure NOMA solution,  \eqref{kktcccc} can be simplified as follows: 
\begin{align}
\left\{\begin{array}{l} 1-\gamma_1\lambda_2  =0 ,
1- \lambda_1 \gamma_2 -\lambda_5=0,\\  
-\frac{\lambda_1 \gamma_0}{1+\epsilon }+\epsilon \lambda_2 \gamma_0  =0, 
 \\  \lambda_3=0,  \lambda_4 =0, \lambda_6 =0
 \end{array}\right.,
\end{align}
which leads to the following closed-form expressions of the Lagrange multipliers: 
\begin{align}
\left\{\begin{array}{l} 
\lambda_1 =\frac{\epsilon (1+\epsilon)}{\gamma_1}  , \lambda_2  =\frac{1}{\gamma_1} ,
\lambda_5=1- \frac{\epsilon (1+\epsilon)}{\gamma_1} \gamma_2, 
 \\  \lambda_3=0,  \lambda_4 =0, \lambda_6 =0.
 \end{array}\right.
\end{align}
In order to guarantee    $\lambda_5\geq 0$,  the following is required:   $1- \frac{\epsilon (1+\epsilon)}{\gamma_1} \gamma_2\geq0$ or equivalently 
$  \frac{\gamma_2}{\gamma_1} \leq \frac{1}{\epsilon (1+\epsilon)}$. Therefore,  the pure NOMA solution,   $P_0=\frac{\epsilon}{\gamma_0}$ and $P_1=\frac{\epsilon(1+\epsilon)}{\gamma_1}$, is optimal, if the following two conditions hold
\begin{align}\label{conditions1}
 \frac{\gamma_2}{\gamma_1} \leq \frac{1}{\epsilon (1+\epsilon)},   \quad \frac{ \gamma_1}{\gamma_0}\leq  (1+\epsilon). 
\end{align}

Second, the    case that   $P_0=P_1=\frac{\epsilon}{\gamma_0}$  is considered, which requires the following assumption:
\begin{align}\label{xdfere}
\frac{\epsilon}{\gamma_0}\geq \frac{\epsilon(1+\epsilon)}{\gamma_1} \Rightarrow \frac{ \gamma_1}{\gamma_0}\geq  (1+\epsilon). 
\end{align}
Note that $P_0=P_1=\frac{\epsilon}{\gamma_0}$  ensures that $R- \log\left(
1+\gamma_0  P_0 
\right)-\log\left(
1+ \gamma_2P_2
\right)=0$ and $\epsilon \gamma_0 P_0+\epsilon-\gamma_1 P_1\neq0$, which implies  that $\lambda_1\neq 0$ and  $\lambda_2= 0$. 
By using the facts that $P_0= P_1\neq 0$ and  $P_2=0$ for the pure NOMA solution,  \eqref{kktcccc} can be simplified as follows: 
\begin{align}
\left\{\begin{array}{l}  \lambda_3=1 ,
1- \lambda_1 \gamma_2 -\lambda_5=0,\\  
-\frac{\lambda_1 \gamma_0}{1+\epsilon} +1 =0,
 \\      \lambda_2 =0,  \lambda_4 =0, \lambda_6 =0
 \end{array}\right.,
\end{align}
which leads to the following closed-form expressions of the Lagrange multipliers: 
\begin{align}
\left\{\begin{array}{l}  \lambda_3=1 ,
\lambda_5=1- \frac{1+\epsilon} { \gamma_0} \gamma_2 ,\\  
\lambda_1 =\frac{1+\epsilon} { \gamma_0},
 \\      \lambda_2 =0,  \lambda_4 =0, \lambda_6 =0
 \end{array}\right.,
\end{align}
To ensure $\lambda_5\geq0$,  $  \frac{\gamma_2} { \gamma_0}\leq \frac{1}{1+\epsilon} $ is required.   It is worth to point out that the condition in \eqref{xdfere} ensures that $\epsilon \gamma_0 P_0+\epsilon-\gamma_1 P_1\leq0$. Therefore,  the pure NOMA solution,   $P_0=P_1=\frac{\epsilon}{\gamma_0}$,  is optimal, if the following two conditions hold
\begin{align}\label{conditions2}
 \frac{\gamma_2} { \gamma_0}\leq \frac{1}{1+\epsilon},   \quad \frac{ \gamma_1}{\gamma_0}\geq  (1+\epsilon). 
\end{align}
Combining \eqref{conditions1} and \eqref{conditions2}, the proof of the lemma is complete. 

\section{Proof of Lemma \ref{lemma4}}\label{proof4}
In this proof, both $P_2$ and $P_0$ are assumed to be non-zero, since both pure NOMA and OMA solutions have already been discussed.  Furthermore, to facilitate performance analysis, the constraints, $P_i\geq 0$, $i\in \{1,2\}$, are omitted, and the corresponding    Lagrangian of the formulated optimization problem  is given by 
\begin{align}\nonumber
L= &P_1+P_2 +\lambda_1\left(
 R- \log\left(
1+\gamma_0  P_0 
\right)-\log\left(
1+ \gamma_2P_2
\right)
 \right)\\ &+\lambda_2\left(\epsilon
 \gamma_0 P_0+\epsilon-\gamma_1 P_1
 \right)+\lambda_3(P_0-P_1)),
\end{align}
where $\lambda_i$'s, $1\leq i \leq 3$, denote the Lagrange   multipliers. 
The use of the stationarity   of the KKT conditions leads to the following:
\begin{align}
1-\gamma_1\lambda_2-\lambda_3&=0,\\\nonumber
1-\frac{\lambda_1 \gamma_2}{1+\gamma_2P_2}&=0,\\\nonumber
-\frac{\lambda_1 \gamma_0}{1+\gamma_0P_0}+\epsilon \lambda_2 \gamma_0+\lambda_3&=0,
\end{align}
which yield   the following conclusions about the transmit power and Lagrange multipliers:   
\begin{align}\label{kkt1}
 P_2&=\lambda_1 -\frac{1}{\gamma_2},\\\nonumber
 P_0&= \frac{\lambda_1 }{\epsilon  \gamma_0\lambda_2 +1-\gamma_1\lambda_2} -\frac{1}{\gamma_0},\\\nonumber
  \lambda_3&=1-\gamma_1\lambda_2.
\end{align}
In order to obtain the insightful understandings about the optimal solution, it is ideal to  remove the Lagrange multipliers from the  transmit power expressions. 

We first note that $\lambda_1\neq 0$, because $\lambda_1= 0$ leads to $P_2=-\frac{1}{\gamma_2}$.  
By using the complementary slackness of the KKT conditions, the assumption of $\lambda_1\neq 0$ means that  \eqref{3tst:1} becomes an equality constraint, i.e., 
\begin{align}
R- \log\left(
1+\gamma_0  P_0 
\right)-\log\left(
1+ \gamma_2P_2
\right)=0. \label{kkt2}
\end{align}
With some straightforward algebraic manipulations, the combination of \eqref{kkt1} and \eqref{kkt2} leads to the following conclusion: 
\begin{align}  \label{station}
  \frac{\gamma_0\gamma_2\lambda_1^2 }{\epsilon \gamma_0\lambda_2 +1-\gamma_1\lambda_2}   =e^R.
\end{align}

\subsection{For the case of   $\lambda_3=0$} By using  \eqref{kkt1} and the assumption of  $\lambda_3=0$, $\lambda_2$ can be expressed as follows: $\lambda_2=\frac{1}{\gamma_1}$. Because $\lambda_2=\frac{1}{\gamma_1}\neq 0$,   \eqref{3tst:3} becomes an equality constraint:  
\begin{align}
&\epsilon \gamma_0 P_0+\epsilon-\gamma_1 P_1= 0 ,
\end{align}
which leads to the following expression of $P_1$:
\begin{align} 
 P_1&= \frac{\epsilon \gamma_0 P_0+\epsilon}{ \gamma_1}. 
\end{align}
In addition, by using \eqref{station} and the fact that $\lambda_2=\frac{1}{\gamma_1}$, $\lambda_1$ can be obtained as follows:
\begin{align}  
  \lambda_1      =\sqrt{\frac{\epsilon e^R}{ \gamma_2\gamma_1}}. 
\end{align}
Therefore, the transmit power solutions can be shown in the following closed-form expressions: 
\begin{align}
 P_0&= \frac{\gamma_1  }{\epsilon \gamma_0}\sqrt{\frac{\epsilon e^R}{ \gamma_2\gamma_1}} -\frac{1}{\gamma_0},
  P_1=   \sqrt{\frac{\epsilon e^R}{ \gamma_2\gamma_1}} ,
 \\\nonumber
 P_2&=\sqrt{\frac{\epsilon e^R}{ \gamma_2\gamma_1}} -\frac{1}{\gamma_2}. 
\end{align}

\subsection{For the case of $\lambda_3\neq 0$} 
Because  $\lambda_3\neq 0$,  \eqref{3tst:3} becomes an equality constraint, i.e.,   $P_0=P_1$, which means that the transmit power solutions are given by
 \begin{align}\label{x1}
 &P_1=P_0= \frac{\lambda_1 }{\epsilon \gamma_0\lambda_2 +1-\gamma_1\lambda_2} -\frac{1}{\gamma_0},\\\nonumber &P_2=\lambda_1 -\frac{1}{\gamma_2}.
\end{align}

\subsubsection{For the case of   $\lambda_2=0$}  the  expressions of $P_1$ and $P_0$ can be simplified as follows:  $ P_1=P_0=\lambda_1 -\frac{1}{\gamma_0}$. In addition, the fact that $\lambda_2=0$ simplifies  the complementary slackness  condition in \eqref{station}, and      can be used  to yield the following expression of $\lambda_1$:
\begin{align} 
  \lambda_1=\sqrt{\frac{e^R}{ \gamma_0\gamma_2}}.
\end{align}
So the transmit power solutions in this case are given by 
\begin{align}
 P_1=P_0&= \sqrt{\frac{e^R}{ \gamma_0\gamma_2}} -\frac{1}{\gamma_0}, P_2=\sqrt{\frac{e^R}{ \gamma_0\gamma_2}} -\frac{1}{\gamma_2}.
\end{align}
We note that evaluating which H-NOMA solution is optimal requires the expressions of all the Lagrange multipliers. It is worth to point out that $  \lambda_3=1-\gamma_1\lambda_2=1$ for this case. 

\subsubsection{For the case of $\lambda_2\neq 0$}   \eqref{3tst:2} becomes an equality constraint, and hence the following equality holds: 
\begin{align}\label{x2}
\epsilon \gamma_0 P_0+\epsilon-\gamma_1 P_1= 0,
\end{align}
which leads to the following observation:
\begin{align}    \label{kkt5}
 P_0= \frac{\epsilon}{\gamma_1-\epsilon \gamma_0}   ,
\end{align}
where the fact that $P_0=P_1$ from \eqref{x1} is used. 

An important observation from \eqref{x1} is that the expression of $P_0$ contains a term of $ \frac{ \lambda_1 }{\epsilon \gamma_0\lambda_2 +1-\gamma_1\lambda_2}$, and this term  can be evaluated by rewriting  \eqref{station} as follows: 
\begin{align}  \label{station2}
  \frac{ \lambda_1 }{\epsilon \gamma_0\lambda_2 +1-\gamma_1\lambda_2}   =\frac{e^R}{\gamma_0\gamma_2\lambda_1}.
\end{align}

Combining \eqref{x1}, \eqref{kkt5} and \eqref{station2}, the following equality can be obtained:
\begin{align}  \label{station4}
 \frac{\epsilon}{\gamma_1-\epsilon \gamma_0}  = \frac{e^R}{\gamma_0\gamma_2\lambda_1} -\frac{1}{\gamma_0},
\end{align}
which means that $\lambda_1$ can be obtained as follows:
\begin{align}  \label{station5}
\lambda_1= \frac{e^R}{ \frac{\epsilon \gamma_0\gamma_2}{\gamma_1-\epsilon \gamma_0} +  \gamma_2} = \frac{e^R(\gamma_1-\epsilon \gamma_0)}{\gamma_1\gamma_2} .
\end{align}

Therefore, the transmit power solutions for this case are given by
\begin{align}
&P_1=P_0=\frac{\epsilon}{\gamma_1-\epsilon \gamma_0}
 ,  P_2=\frac{e^R(\gamma_1-\epsilon \gamma_0)}{\gamma_1\gamma_2} -\frac{1}{\gamma_2}.
\end{align}
The closed-form expressions of the Lagrange  multipliers are important to evaluate which form of H-NOMA is optimal. For this case, the Lagrange  multipliers can be obtained as follows:
\begin{align}  
 \lambda_2    =&\frac{\gamma_0\gamma_2\lambda_1^2}{e^R(\epsilon \gamma_0-\gamma_1)}-\frac{1}{\epsilon \gamma_0-\gamma_1}\\\nonumber 
 =&  \frac{\gamma_0e^{R}(\epsilon \gamma_0-\gamma_1)}{\gamma_1^2\gamma_2}  -\frac{1}{\epsilon \gamma_0-\gamma_1},
\end{align}
and
\begin{align}  
  \lambda_3   =&1-\gamma_1\lambda_2=1- \frac{\gamma_0e^{R}(\epsilon \gamma_0-\gamma_1)}{\gamma_1\gamma_2}  +\frac{\gamma_1}{\epsilon \gamma_0-\gamma_1}.  
\end{align}
The proof of the lemma is complete. 
\bibliographystyle{IEEEtran}
\bibliography{IEEEfull,trasfer}
  \end{document}